\documentclass[draftcls,onecolumn,12pt,oside]{IEEEtran}
\usepackage{xcolor}

\usepackage{color,soul}

\usepackage{cite}

\ifCLASSINFOpdf
   \usepackage[pdftex]{graphicx,epsfig}
\else
   \usepackage[dvips]{graphicx}
\fi
\usepackage{amssymb}

\usepackage[cmex10]{amsmath}
\usepackage[subnum]{cases}

\usepackage{algorithm}
\usepackage{algorithmic}

\usepackage{epstopdf}

\newcommand{\argmax}{\operatornamewithlimits{arg\,max}}
\newcommand{\argmin}{\operatornamewithlimits{arg\,min}}
\ifCLASSOPTIONcompsoc
\usepackage[tight,normalsize,sf,SF]{subfigure}
\else
\usepackage[tight,footnotesize]{subfigure}
\fi

\graphicspath{C:/Users/minh/Desktop/Image_BDD/}

\newtheorem{theorem}{Theorem}[section]
\newtheorem{lemma}[theorem]{Lemma}

\newtheorem{remk}[theorem]{Remark}
\newtheorem{propos}[theorem]{Proposition}
\newtheorem{coroll}[theorem]{Corollary}

\usepackage{footnote}
\usepackage{tikz}
\usetikzlibrary{shapes,arrows}
\usepackage{stfloats}
\begin{document}
%
\title{On Optimal Latency of Communications}
%
%
%

\author{Minh~Au,~\IEEEmembership{Member,~IEEE}
		Fran\c cois~Gagnon,~\IEEEmembership{Senior~Member,~IEEE},
\thanks{M. Au is a Postdoctoral fellow at Department of Electrical Engineering 
at \'Ecole de technologie sup\'erieure, Montr\'eal,
QC, H3C1K3 e-mail: Minh.Au@lacime.etsmtl.ca.}
\thanks{F. Gagnon is a Professor at Department of Electrical Engineering 
at \'Ecole de technologie sup\'erieure, Montr\'eal,
QC, H3C1K3 e-mail: Francois.Gagnon@etsmtl.ca.}
}

%
%

\markboth{}%
{Au \MakeLowercase{\textit{et al.}}:  }
%



\maketitle

\begin{abstract}
In this paper we investigate the optimal latency of communications. Focusing on fixed rate communication without any feedback channel, this paper encompasses low-latency strategies with which one hop and multi-hop communication issues are treated from an information theoretic perspective. By defining the latency as the time required to make decisions, we prove that if short messages can be transmitted in parallel Gaussian channels, for example, via orthogonal frequency-division multiplexing (OFDM)-like signals, there exists an optimal low-latency strategy for every code. This can be achieved via early-detection schemes or asynchronous detections. We first provide the optimal achievable latency in additive white Gaussian noise (AWGN) channels for every channel code given a probability block error $\epsilon$. This can be obtained via sequential ratio tests or a ``genie'' aided, \textit{e.g}. error-detecting codes. Results demonstrate the effectiveness of the approach. Next, we show how early-detection can be effective with OFDM signals while maintaining its spectral efficiency via random coding or pre-coding random matrices. Finally, we explore the optimal low-latency strategy in multi-hop relaying schemes. For amplify-and-forward (AF) and decode-and-forward (DF) relaying schemes there exist an optimal achievable latency. In particular, we first show that there exist a better low-latency strategy, for which AF relays could transmit while receiving. This can be achieved by using amplify and forward combined with early detection.
\end{abstract}


%
\IEEEpeerreviewmaketitle

\section{Introduction}
\IEEEPARstart{A}{s} more devices are developed and interconnected, the communication infrastructure becomes critical and time-sensitive for many applications such as drone control, remote surgery, and vehicle to vehicle communications. These critical communications must be resilient to interference, jamming, and intrusions and are often required to provide services with extremely low latency. Obviously, a speed of light induced latency is necessarily present. For communicating over a one thousand kilometres, this unavoidable latency is around $3$ ms. Interestingly, wireless communication is faster than optical fiber guided communication due to a larger index of refraction in the optical fiber. Over long distance ranges, wireless relaying schemes could be an efficient solution for low-latency communication. Several works has been dedicated to this purpose \cite{Bradford2012,Elson2007,Yi2007,Neely2005}. To our knowledge, these research works might not be optimal for low-latency communication. Hence, this work aims to analyse the optimal low-latency communication strategies in such contexts.

Shannon's fundamental limit on communication bounds the capacity of a communication medium to transmit information \cite{Shannon1948}. Unfortunately, it does not bound the time required for a receiver to make a decision on a received symbol. Recent results on the acheivability of the channel capacity in the finite-blocklength regime also impose a minimal duration for the transmission \cite{Polyanskiy2010} under synchronous detection, \textit{i.e}. decisions at a fixed sampling period. Indeed, given a probability of block error $\epsilon$, the minimal blocklength needed to achieve a fraction of the capacity $\eta$ is given by:
\begin{equation}
n \approx \left(\frac{Q^{-1}(\epsilon)}{1-\eta} \right)^{2}\frac{V}{C^{2}}
\label{eq:0.0.1} 
\end{equation}
where $Q^{-1}(\cdot)$ is the inverse of the Q-function and $V$ is the dispersion of the channel. For example, for an AWGN channel with a signal to noise ratio (SNR) $= 20$ dB, $\epsilon = 10^{-6}$, and $\eta = 0.9$, the minimal blocklength is $n \approx 190$ symbols. Since the channel coding rate is $R = k/n = \eta C$, the maximal information block size is $k \approx 316$ bits. However, it is not clear that such a channel code with $k = 316$ bits and blocklength $n = 190$ symbols is known. On the other hand, without prior knowledge of the channel behaviour, feedback communications and rateless coding schemes have been proposed to adapt opportunistically to channel variations which provide reliable communications \cite{Lomnitz2013,Shayevitz2009,Lomnitz2011,Draper2009,Draper2004}. Notably, \cite{Draper2004} have proposed a strategy for a reliable communication over an unknown channel by testing periodically the received sequence. As soon as a message can be decoded, an acknowledgement is sent which stop the transmission. Unfortunately, neither these strategies provide the minimal duration for the transmission nor do they prove that they are optimal in terms of latency.
 
Furthermore, extensive research works in \cite{Neely2005,Musavian2010,Elson2007,Yi2007,Kanodia2001} show that the issue of latency is central to the multi-hop relaying problem. The general relaying problem has been studied and it has many variations as presented in \cite{Cover1979,Neely2005}. Ideally, when the messages are long, one can reduce latency by using decode-and-forward (DF) schemes by dividing the message into smaller parts. However, the receiver must wait an entire codeword before determining any of the input symbols which yields greater computational complexity and latency \cite{Bradford2012}. On the other hand, if one does not have a receive power or error probability issue, an amplify-and-forward (AF) scheme has relatively low complexity and minimal latency \cite{Bradford2012}. 

A first question that can be explored concerns the relative merit of each scheme under realistic power, bandwidth, and error probability requirements. One could thus obtain the channel condition regions which is optimal for each case. This however avoids the main issue: if neither is always optimal, what is the overall optimal low-latency scheme under realistic channel constraints? To answer this question, we will first explore both extremes and obtain their specific latency. 

In this paper, we aim to explore low-latency communication through fixed rate channel codes wth no feedback and extended to multi-hop systems. Furthermore, we assume that all symbols of a message are transmitted simultaneously in parallel over the channel. As fundamental results, we extend results obtained by \cite{Polyanskiy2010} on the achievability bound, wherein the optimal latency message size is determined for a given usable bandwidth and received power. In addition, in an optimistic channel condition, we first show that the latency problem is approximately linked to the bandwidth and power, in which case if all symbols are transmitted through a parallel channel via orthogonal frequency division multiplexing (OFDM)-like signals, then the latency can be reduced by using early-detection schemes. 

The second fundamental result shows that it is possible to reformulate AF relaying schemes to attain a shorter latency than AF schemes and/or DF schemes only. This can be obtained through early detection strategies or asynchronous detections. For such applications, sequential detection can minimize the time needed to make a decision on a received symbol. Thus, the transmitter can send symbols of longer duration to ensure their receivablility, while the receiver can make a decision on that symbol as soon as the probability of its correct detection is high enough. 

Consider a good channel code whose $\epsilon$ is the solution to equation (\ref{eq:0.0.1}) for a coding rate of $R = 0.5$ bit per channel use and a blocklength of $n = 150$ symbols. All symbols are simultaneously transmitted on parallel independent AWGN channels, and has a fixed symbol duration denoted by $T$. The performance of such a code is depicted in Fig. \ref{fig:0.1.1}.  
\begin{figure}[!htbp]
\centering
\includegraphics[width=3.0in,height=2.8in]{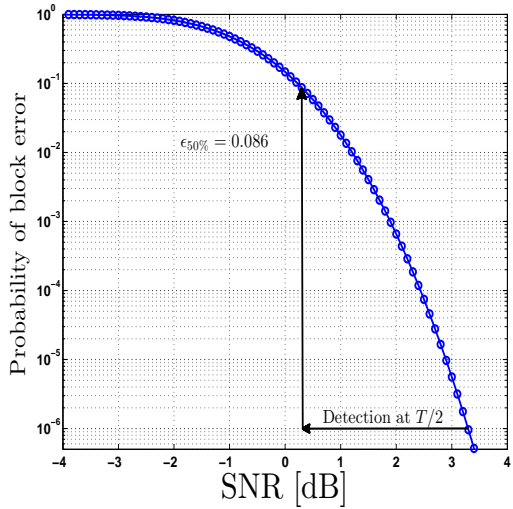}
\caption{Early detection strategy for messages that use a good channel code: Channel coding rate $R = 0.5$ bit per channel use and blocklength $n = 150$ symbols}
\label{fig:0.1.1}
\end{figure}
To reduce latency, let us consider that the receiver performs an early detection at $\tau = \left\lbrace T/2,T\right\rbrace$ in which we could stop as soon as the probability of its correct detection is high enough. The average latency is determined as follows: with respect to a maximal probability of block error $\epsilon(T) = 10^{-6}$, we wish to early detect the message. At $T/2$, it follows that messages can be received successfully with an error probability of $\epsilon(T/2) = 0.086$. The average latency is determined by the expectation of having correct decisions at both $T/2$ and $T$. To do so, assuming the decoder chooses the correct message, the probabilities of having a correct decision at $T/2$ and $T$ are respectively:
\begin{subequations}
\begin{align}
p(T/2) = 1-\epsilon(T/2)\\
p(T) =  (1-\epsilon(T))-p(T/2)
\end{align}
\label{eq:0.0.2}
\end{subequations}
Using equation (\ref{eq:0.0.2}), we have $p(T/2)= 0.914$ and $p(T) = 0.086-10^{-6}$. Therefore, by the expectation value of $\tau$, the message could be received successfully by using only $\approx 54.3\%$ of the duration of symbols on average. In other words, early-detection schemes reduce latency especially when the propagation delay is small compared to the duration of the data transmission. 

The remainder of this paper is organized as follows. In Section \ref{sect:I}, we formulate the minimal-latency problem to investigate the optimal achievable latency from an information theoretical perspective. This allows for the identification and the characterization of channel codes for low-latency communication. It can be shown that for a high signal-to-noise ratio, reducing latency is linked to the bandwidth and power. In such a case, the communication latency should not be defined as the duration of data transmission, but by the difference between the beginning of the transmission of a message and the instant of the correct decision. 

In Section \ref{sect:II}, we define early-detection schemes that can be employed for such a purpose. We first derive the minimal acheivable latency when the receiver performs early-detection schemes. We prove that, if all symbols of the message are transmitted simultaneously in parallel through AWGN channels, then early-detection schemes can significantly reduce latency for short message length. We design simple efficient low-latency communication schemes. In particular, sequential detection tests can be used for reducing the time required to make a decision quickly. In literature, sequential probability ratio tests have been proposed to minimize the time required to make a decision for which the probabilities of error do not exceed predefined values \cite{Wald1948,Chernoff1959,Lehmann1959,Baum1994,Dragalin1999,Dragalin2000,Poor2009}. Unfortunately, messages consist of a large number of symbols, which renders these tests computationally prohibitive due to a large number of possible messages. On the other hand, it is possible to use a list of the $\ell$ most probable codewords or messages in order to reduce the number of hypothesis on the message that has been sent. 

In Section \ref{sect:III}, two examples employing these proposed fastest-detection schemes are discussed under various channel conditions and message lengths. It can be shown that these early-detection schemes minimize the time required to make decisions on messages compared to decisions at the end of the symbol duration. This allows for reducing latency. We then discuss the use of early-detection schemes in OFDM. We prove that one can minimize the time required to make decisions by using random coding schemes, or when codewords are pre-coded by random rotation or orthogonal matrices. 

In Section \ref{sect:IV}, we explore the minimal latency for multi-hop communications. We investigate the optimal scheme which reduces latency in relaying schemes under various channel conditions and message lengths. For DF relaying schemes, it is possible to reduce latency by using early-detection schemes. However, dividing short messages into smaller parts does not reduce latency because the receiver has to wait for the last part of the message to make its decision. Furthermore, the absence of short codes makes the division less attractive. Focusing on short messages and considering early-detection in AF relaying scheme, it is possible to reduce latency by detecting before the end of the symbol duration. Section \ref{sect:V} summarizes our main findings. 

\section{On the optimal achievable latency for synchronous detection schemes}
\label{sect:I}
\subsection{Minimal Latency Problem}
The minimal latency problem can be formulated as the extension of the classical information theory channel capacity equation. Lets consider $M$ possible messages encoded by a good channel code and a codeword is denoted by $\mathbf{X} \in \mathbb{R}^{n}$. The received vector is then given by:
\begin{equation}
\mathbf{Y} = \mathbf{X}+\mathbf{N},
\end{equation}
where $\mathbf{N} \sim \mathcal{N}(0,\mathbf{I}_{n})$ is multidimensional Gaussian noise, independent and identically distributed (i.i.d.) individual components with zero mean and unit variance. Then, let $s(t)$ be the signal being sent such that:
\begin{equation}
\begin{split}
s(t) &= \sum\limits_{i = 1}^{n} X_{i} \varphi_{i}(t),
\end{split}
\label{eq:1.1.00}
\end{equation}
where $\varphi_{i}(t)$ is an orthonormal basis spanning the vector-space of signals used for the transmission of the input sequence $\mathbf{X}$ with a duration $T$. The received signal $y(t)$ is simply given by:
\begin{equation}
y(t) = s(t)+n(t),
\end{equation}
and $\mathbf{Y}$ is obtained by projection onto the orthogonal basis. Furthermore, for bandpass signals, the minimal $3$ dB bandwidth required to send signals is $W = n/2T$. Then, assuming $\mathbf{X}$ satisfies an input constraint $\rho$ such that:
\begin{equation}
\Vert \mathbf{X}\Vert^{2} = n\rho,
\label{eq:1.1.1}  
\end{equation}
where $\rho  = PT$ and $P$ is the received power. Under such a condition, Shannon's capacity in bit per channel use for an AWGN channel is given by:
\begin{equation}
C = \frac{1}{2} \log_{2}\left( 1+\rho\right)
\label{eq:Io.B.1bis} 
\end{equation}
The capacity theorem asserts that $\log_{2}(M)$ cannot be greater than $C$ for reliable transmission. Since the incurred latency is $T$, the theory is simple: the latency is reduced by increasing the bandwidth or the power. Specifically, for a given power and bandwidth, the minimal latency is simply the number of bits to transmit divided by the capacity:
\begin{equation}
L_{\min} = \frac{\log_{2}(M)}{C}T
\end{equation}
Unfortunately, the coding theorem is stated over a very long message. Nevertheless, the latest results in coding theorem in the finite-blocklength regime, states that there exists a coding scheme for which the maximal code size achievable with a given error rate $\epsilon$ and a blocklength $n$ can be bounded for various channel conditions \cite{Polyanskiy2010}. 

\subsection{Minimal Achievable Latency in the Finite-Blocklength Regime}
Consider a code with a blocklength $n$ and $M$ codewords. We denote $\mathbf{X}$ as the input sequence encoded by an ($n$,$M$) code. $\mathbf{Y}$ is the corresponding output sequence induced by $\mathbf{X}$ via a channel $P_{\mathbf{Y}\mid \mathbf{X}}:\mathcal{A}^{n} \rightarrow \mathcal{B}^{n}$, which is a sequence of conditional probabilities where $\mathcal{A}$ and $\mathcal{B}$ are the input and output alphabets respectively \cite{Verdu1994}. An ($n$,$M$) code has an encoding function $f : \left\lbrace 1,\cdots, M \right\rbrace \rightarrow \mathcal{A}^{n}$ and a decoding function $g : \mathcal{B}^{n} \rightarrow  \left\lbrace 1,\cdots, M \right\rbrace$. By definition, if the messages are equiprobable a priori, the average probability of error is given by:
\begin{equation}
P_{e} = \frac{1}{M}\sum\limits_{m = 1}^{M}\text{P}_{\mathbf{Y}\mid m}(g(\mathbf{Y})\neq m)
\end{equation}
where the measure $\text{P}_{\mathbf{Y}\mid m}$ denotes the conditional probability that the decoder $g(\mathbf{Y})$ has made a correct or an incorrect decision on the message, when the actual message $m$ was transmitted. An ($n$,$M$) code whose average probability of error is not larger than $\epsilon$ is an ($n$,$M$,$\epsilon$) code. For a real-valued discrete-time AWGN channel, with
\begin{itemize}
\item $\mathcal{A} = \mathbb{R}$
\item $\mathcal{B} = \mathbb{R}$
\item $P_{\mathbf{Y}\mid \mathbf{X} = \mathbf{x}}  = \mathcal{N}(\mathbf{x},\mathbf{I}_{n})$ 
\end{itemize}
The maximal code size achievable with a given probability of error and blocklength $M^{*}(n,\epsilon)$ is given by the following theorem \cite{Polyanskiy2010}.
\begin{theorem}[Polyanskiy \textit{et al.}]
For the AWGN channel with SNR $\rho$, $0<\epsilon < 1$, and for equal-power, maximal-power and average-power constraints, the maximal code size achievable $\log_{2} M^{*}(n,\epsilon)$ is given by:
\begin{equation}
\log_{2} M^{*}(n,\epsilon) = nC-\sqrt{n V} Q^{-1}(\epsilon) + O(\log_{2} n)
\label{eq:II.B.1}
\end{equation}
where $V$ is the channel dispersion such that:
\begin{equation}
V = \frac{\rho}{2} \frac{\rho+2}{(\rho+1)^{2}} \log_{2}^{2}e
\label{eq:II.B.1bis}
\end{equation}
\end{theorem}

Specifically, there exists a coding scheme such that the achievability bound for equal-power and maximal power constraints such that:
\begin{equation}
\begin{split}
\log_{2} M_{e,m}^{*}(n,\epsilon) \leq nC-\sqrt{nV} Q^{-1}(\epsilon) + \frac{1}{2}\log_{2} n + O(1)
\end{split}
\label{eq:II.B.2bis}
\end{equation}
whereas for average-power constraint, we have:
\begin{equation}
\begin{split}
\log_{2} M_{a}^{*}(n,\epsilon)\leq nC-\sqrt{nV} Q^{-1}(\epsilon) + \frac{3}{2}\log_{2} n + O(1)
\end{split}
\label{eq:II.B.2bis2}
\end{equation}
\begin{remk}
Under a normal approximation, an explicit result in terms of physical variables that are linked to the channel can be obtained by introducing the latency $L$, the power $P$ and the symbol duration $T$ in equations (\ref{eq:II.B.2bis}) and (\ref{eq:II.B.2bis2}). Moreover, by assuming $T = 1/2W$ and $L = nT$, the achievability bound is rewritten as:
\begin{equation}
\begin{split}
\log_{2} M^{*}(n,\epsilon)\leq &\frac{L}{2T}\log_{2}(1+PT)-\sqrt{\frac{L}{T} \frac{PT(PT+2)}{2(PT+1)^{2}}}\\
&\cdot\log_{2}(e)Q^{-1}(\epsilon)+ \frac{1}{2}\log_{2}\left(  \frac{L}{T}\right)+O(1) 
\end{split}
\label{eq:II.B.2biss}
\end{equation}   
\end{remk}

\begin{remk}
By approximating the logarithmic functions, two observations are possible: when $PT$ is small, equation (\ref{eq:II.B.2biss}) may be rewritten for a power-limited region:
\begin{equation}
\begin{split}
\log_{2} M^{*}(n,\epsilon)\sim &\frac{LP}{2}-\sqrt{LP}\log_{2}(e)Q^{-1}(\epsilon)\\
& + \frac{1}{2}\log_{2}\left(  \frac{L}{T}\right)+O(1) 
\end{split}
  \end{equation}  
Under such a condition, the first two terms are independent of $T$ and thus, increasing the bandwidth indefinitely does not work well. In particular, it occurs when $PT\ll 1$. On the other hand, when $PT\gtrsim 2$, equation (\ref{eq:II.B.2biss}) may be rewritten for a bandwidth-limited case:
\begin{equation}
\begin{split}
\log_{2} M^{*}(n,\epsilon)\sim &\frac{L}{2T}\log_{2}(1+PT)-\sqrt{\frac{L}{T}}\log_{2}(e)Q^{-1}(\epsilon)\\
&+ \frac{1}{2}\log_{2}\left(  \frac{L}{T}\right)+O(1) 
\end{split}
  \end{equation} 
As for the classical capacity equation, reducing the latency is approximately linked to the bandwidth and power in this case.  
\end{remk}

For a fixed duration transmission, the minimal latency is simply obtained by using equation (\ref{eq:II.B.2biss}). For a single hop wireless link, and assuming a known constant channel gain where $P$ is the received power, the minimal latency is depicted in Fig. \ref{fig:0.1.2.1} for an error rate of $\epsilon = 10^{-7}$. For example, let us consider that we have $PT = 2.5$ or $4$ dB, and an information block size message of $k = 103$ bits needs to be sent. If all the bandwidth is used, then a blocklength of $n = 186$ symbols or channel use are necessary and for a bandwidth of $W = 50$ MHz, the latency would be $L = 1.86~\mu$s. However, to our knowledge, we don't know of such a good ($186$,$2^{103}$,$\epsilon$) code.  
\begin{figure*}[!htbp]
\centering
\includegraphics[width=5.0in,height=3.8in]{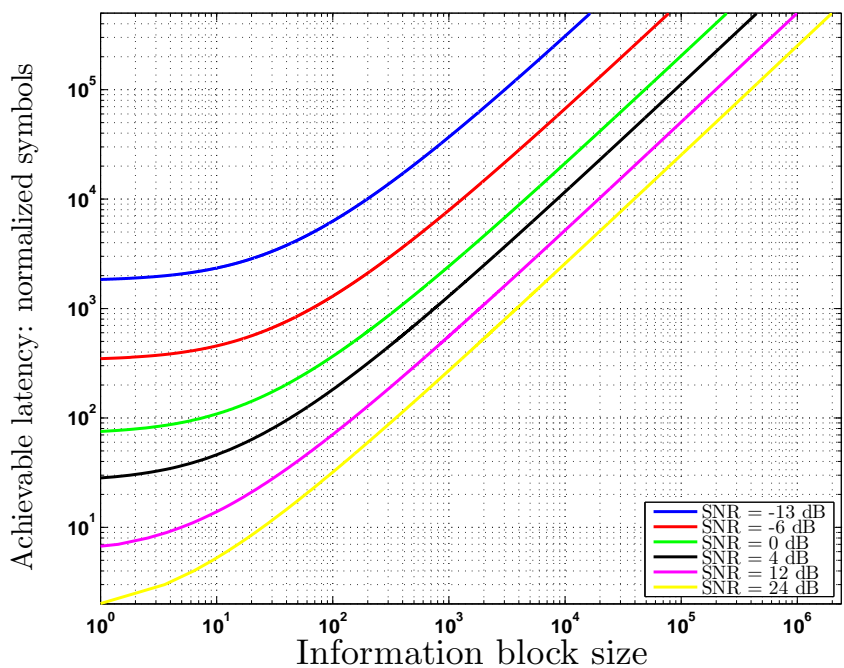}
\caption{Achievable latency as a function of the information block size message for different channel conditions, probability of block error  $\epsilon = 10^{-7}$}
\label{fig:0.1.2.1}
\end{figure*}
\begin{remk}
Another noteworthy feature of this bound is the quasi-linearity for most power-bandwidth ratios of interest for about $k = 1000$ bits. By deriving the bound in $n = L/T$, we obtain:
\begin{equation}
\begin{split}
\frac{d\log_{2} M^{*}(n,\epsilon) }{d L/T} \sim& \frac{1}{2}\log_{2}(1+PT)+\frac{1}{2\sqrt{L/T}}\\
&\cdot\sqrt{\frac{PT(PT+2)}{2(PT+1)^{2}}}\log_{2}(e)Q^{-1}(\epsilon)+ \frac{1}{2L/T} 
\end{split}
\end{equation}
Thus, as the code length grows, there is not much advantage in terms of extra information block size that are encoded. In fact, at about ten thousand symbols, at reasonable channel conditions and error probabilities, the effective code rate can improve by $1\%$. Focusing on code lengths of about $100$ to $1000$ should be paramount. The linearity is also important because the latency in multi-hop communication can be reduced by breaking up messages into smaller parts. 
\end{remk}

Most of the theory has until now been concerned with the choice of transmitter characteristics assuming the maximum likelihood receiver. Unfortunately, analysing the latency is not sufficient. Communication latency should not be defined as the duration of data transmission, but by the difference between the beginning of the transmission of a message and the instant of correct decision. Ignoring the propagation delay, the latency can be smaller than the symbol duration.

\section{Early-detection using sequential tests}
\label{sect:II}
\subsection{Problem Formulation}
To reduce latency, we should consider such a scenario, in which a message is sent over a long period of time. The receiver would use the minimal proportion of time to compute the required message, hence we discount the detection or decoding computation. We formulate the early-detection problem as follows. Assuming a transmitter sends a message through a channel with a fixed symbol duration $T$ and let $\mathbf{Y} = \left\lbrace \mathbf{Y}_{1},\mathbf{Y}_{2}, \cdots, \mathbf{Y}_{T}\right\rbrace $ be a collection of random vectors, i.i.d. with a density $\text{P}_{t}$, whose $\mathbf{Y}_{t}$ is the received message signal sampled at $t\leq T$. In particular, let consider: 
\begin{equation}
\begin{split}
\mathbf{Y}_{t} &= \sum\limits_{i = 1}^{t} \mathbf{Y}_{i}\\
\end{split}
\label{eq:3.bis.1}
\end{equation}
where $\mathbf{Y}_{i}$ is given by:
\begin{equation}
\begin{split}
\mathbf{Y}_{i} &= \mathbf{X}_{i}+\mathbf{N}\\
\end{split}
\label{eq:3.bis.2}
\end{equation}
where $\mathbf{N} \sim \mathcal{N}(0,\mathbf{I}_{n})$ is the additive white Gaussian noise vector. $\mathbf{X}_{i} \in \mathbb{R}^{n}$ is the message signal of a blocklength $n$ indexed by $i$ which satisfies the equal power constraint $\rho_{i}  = P\tau_{i}$, where $\tau_{i}$ is a small proportion of the symbol duration:
\begin{equation}
\Vert \mathbf{X}_{i}\Vert^{2} = n\rho_{i}
\label{eq:3.bis.3} 
\end{equation}
and specifically, when $\rho = PT $, we have: 
\begin{equation}
\sum\limits_{i = 1}^{T}\Big\Vert \mathbf{X}_{i}\Big\Vert^{2} = n\rho  
\end{equation}
In this paper we assume that the received signal has been projected onto the orthogonal basis to obtain $\mathbf{Y}_{i}$. This means that the bases spanning the vector-space of signals are orthonormal for all the proportions of the symbol duration $\tau_{i}$. Hence, all components in $\boldsymbol\varphi (t)$ of a duration $\tau_{i}$ are orthonormal.     

To perform low-latency communication, a sequential test signal detection can be used. Assume that there is $M$ possible messages that are equiprobable a priori. Thus, the optimal decision about which message was transmitted with respect to a finite number of samples is formulated as a stopping time \cite{Chernoff1959,Lehmann1959,Siegmund1985}
\begin{equation}
\tau_{m} = \inf\left\lbrace t: \hat{m} = \argmax\limits_{1\leq m\leq M} \text{P}_{t}(\mathbf{Y}_{t}|m)> S_{m}  \right\rbrace 
\label{eq:I.1.1}
\end{equation}
In other words, we stop as soon as the maximum likelihood exceeds a threshold $S_{m}$, decide that $m$ was transmitted. Since the sequence $\mathbf{Y}_{1},\mathbf{Y}_{2}, \cdots, \mathbf{Y}_{T}$ are i.i.d., $\text{P}_{t}(\mathbf{Y}_{t}|m) $ can be written as:
\begin{equation}
\text{P}_{t}(\mathbf{Y}_{t}|m) = \prod\limits_{i=1}^{t} \text{P}(\mathbf{Y}_{t}|m) 
\end{equation}
Under such a condition, the latency can be reduced by deciding as soon as the message $m$ was transmitted. Indeed, if one can decide whether the message has been received correctly at $\tau_{m} \leq T$, then the latency is $n\tau_{m} \leq nT$, and even more when $\tau_{m}$ is smaller. Such a scheme is an early detection in which the receiver is required to wait until the probability of having a message in the codebook is high enough to make a decision.

\subsection{Minimal Achievable Latency using an Optimal Early-detection Scheme}
An early-detection scheme is optimal when it minimizes the expectation of the proportion of the symbol duration used for detection, for which the error rate does not exceed a predefined value. Moreover, we assume that a perfect error-detecting code (no false positive) checks whether or not, there are errors in the transmitted message. The decision rule is straightforward: as soon as an error has not been detected, decode the message. Otherwise, wait for the next sample until no error occurs. We should note that if an error is still detected until the end of the transmitted symbol $T$, the decision will be made at $T$ by default. 

Under such a condition, the average latency can be simply determined by the error rate as a function of the SNR $\rho$. As an example, we use an arbitrary ($n$,$M$,$\epsilon$,$\rho$) code whose average probability of decoding error over the channel is not larger than $\epsilon$. By using equation (\ref{eq:II.B.1}) and under the normal approximation, we obtain the maximum achievable channel coding rate $R^{*}(n,\epsilon,\rho)$ (in bits per channel use) such that:
\begin{equation}
R^{*}(n,\epsilon,\rho) \approx C-\sqrt{\frac{V}{n}}Q^{-1}(\epsilon)+\frac{1}{2n} \log_{2} n
\label{eq:II.B.2}
\end{equation}
It follows that for a given fixed channel coding rate $R$, an SNR $\rho$, and a blocklength $n$, the performance in terms of error rate of such codes can be determined using equation (\ref{eq:II.B.2}) by:
\begin{equation}
\epsilon^{*}(\rho,R,n) \approx Q\left( \frac{C-R+ \frac{1}{2n}\log_{2} n}{\sqrt{V/n}}\right) 
\label{eq:II.B.3}
\end{equation}
The performance of such ($n$,$M$,$\epsilon$) codes in the finite-blocklength regime is depicted in Fig. \ref{fig:0.1.2}, for $R = 0.5$ and $R = 0.95$ in bit per channel use as well as for various blocklength $n$.  
\begin{figure}[!htbp]
\centering
\includegraphics[width=3.7in,height=3.0in]{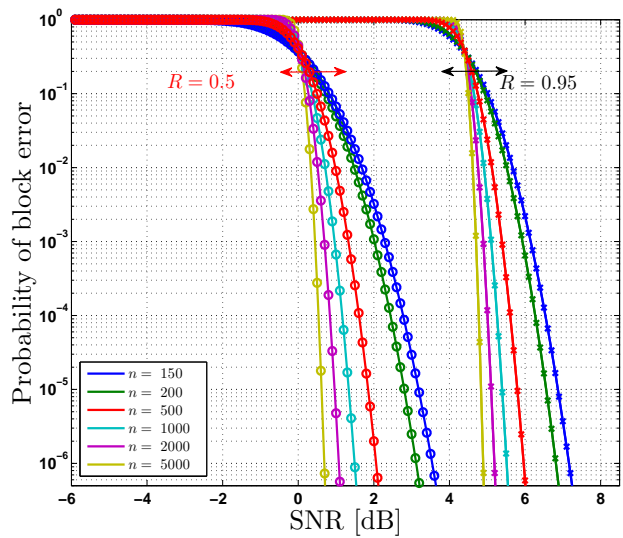}
\caption{Performance of ($n$,$M$,$\epsilon$) codes in the finite-blocklength regime}
\label{fig:0.1.2}
\end{figure}
\begin{remk}
For very long messages and $R > C$, reliable communication is not possible which renders the sequential test inefficient in the sense that the error rate is close to one, and in fact, such a test would never satisfy equation (\ref{eq:I.1.1}), even though it is optimal. In contrast to the latter, for $R < C$ as long as the error rate is smaller than one, the minimal proportion of the symbol duration required to obtain a message, can be provided by Shannon's capacity formula such that: 
\begin{equation}
 \frac{2^{2R}-1}{P}< \inf(\tau).
\end{equation}
\end{remk}
From this remark, we conclude that latency has a strong dependence on the behaviour of the error probability under a given range of SNR.

\begin{theorem}
The optimal average latency of early detection schemes for an arbitrary ($n$,$M$,$\epsilon$,$\bar{\tau}$) is given by:
\begin{equation}
\bar{\tau}  \leq  \int_{0}^{T} -\tau dQ(\gamma(\tau)) d\tau
\label{eq:3.3.5}
\end{equation}
where $\gamma(\tau) =  \frac{C(P\tau)-R+ 1/2n \log_{2} n}{\sqrt{V(P\tau)/n}}$ and $dQ(\gamma(\tau))$ is the differential of the $Q$-function given by equation (\ref{eq:II.B.3}) for all $\gamma(\tau)$.
\label{th:1}
\end{theorem}
This theorem can be proved by using the following lemma.  
 
\begin{lemma}
For an AWGN channel with an arbitrary ($n$,$M$,$\epsilon$,$\bar{\tau}$) code whose error probabilities satisfy equation (\ref{eq:II.B.3}) and for all $d \tau > 0$, the distribution of $\tau$ for an optimal early-detection scheme is given by the differential of the error:
\begin{equation}
\begin{split}
\lim\limits_{d \tau \rightarrow 0} p(\tau+d \tau) &= -\frac{d Q(\gamma(\tau))}{d\gamma(\tau)}d\gamma(\tau)\\
&= \frac{1}{\sqrt{2\pi}}e^{-\gamma(\tau)^{2}/2}d\gamma(\tau) 
\end{split}
\label{eq:3.3.3}
\end{equation}
\end{lemma}

\begin{proof}
Consider two possible events $\tau_{1}$ and $\tau_{2}$ where $\tau_{1} < \tau_{2} \leq T$ where a perfect error-detecting code checks whether or not, there are errors in the transmitted message. In this case, the probability of having a correct decision at $\tau_{1}$ is given by $p(\tau_{1}) = 1-\epsilon^{*}(P\tau_{1},R,n)$. If the decoder has not decided at $\tau_{1}$ it means that there is an error in the message and we wait for the next sample at $\tau_{2}$. The immediate consequence is that the probability to make a correct decision at $\tau_{2}$ depends on the probability of having a correct decision previously. Indeed, if the decoder has decided at $\tau_{2}$, then it means that errors is detected and errors have been detected previously. Therefore, the probability of having a correct decision at $\tau_{1}$ is $p(\tau_{2}) = (1-\epsilon^{*}(P\tau_{2},R,n))-(1-\epsilon^{*}(P\tau_{1},R,n)) = \epsilon^{*}(P\tau_{1},R,n)-\epsilon^{*}(P\tau_{2},R,n)$. Fig. \ref{fig:2.1.2.bb} illustrates the decision rule in an optimal early-detection scheme using a perfect error-detecting code.

\begin{figure}[!htbp]
\centering
\includegraphics[width=2.7in,height=2.0in]{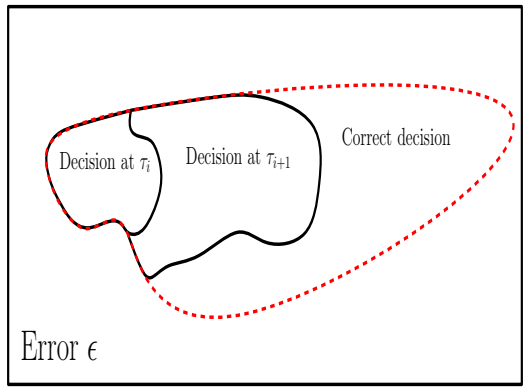}
\caption{Decision in an optimal early-detection scheme using a perfect error-detecting code}
\label{fig:2.1.2.bb}
\end{figure}

Hence, we could thus generalize our purpose by letting $S_{\tau} = \left\lbrace \tau_{1}, \tau_{2},\cdots  T\right\rbrace $ be an increasing and positive sequence of samples that are used to make a decision on the message. In addition, we consider that there exists a $\tau + d \tau \in S_{\tau}$ such that the probability of having a correct at $\tau + d\tau$ is given by:
\begin{equation}
\begin{split}
p(\tau+d\tau) &= (1-\epsilon^{*}(P\tau,R,n))-(1-\epsilon^{*}(P(\tau+d\tau),R,n)\\
&= \epsilon^{*}(P\tau,R,n)-\epsilon^{*}(P(\tau+d\tau),R,n)
\end{split}
\label{eq:3.3.3.b}
\end{equation}
%
Furthermore, let the capacity $C(P\tau)$ and the channel dispersion $V(P\tau)$ be two time-dependent functions where $\rho$ in equations (\ref{eq:Io.B.1bis}) and (\ref{eq:II.B.1bis}) are replaced by $P\tau< P(\tau+d\tau) \leq PT$. By using equation (\ref{eq:II.B.3}), we define $\gamma(\tau)$ as: 
\begin{equation}
\gamma(\tau) = \sqrt{n}\frac{C(P\tau)-R+ \frac{1}{2n}\log_{2} n}{\sqrt{V(P\tau)}}
\label{eq:3.3.3.c}
\end{equation}
Therefore, since for any ($n$,$M$,$\epsilon$,$\bar{\tau}$)-code whose satisfies equation (\ref{eq:II.B.3}), the distribution of $\tau$ can be simplified as in equation (\ref{eq:3.3.3}) by using equations (\ref{eq:3.3.3.c}) and (\ref{eq:II.B.3}) in (\ref{eq:3.3.3.b}) by letting $d\tau \rightarrow 0$. The average latency of an ($n$,$M$,$\epsilon$,$\bar{\tau}$) code is given by the expectation of $\tau$ as in equation (\ref{eq:3.3.5}). This concludes the proof. 
\end{proof}

We provide an analysis of such codes whose latency can be reduced using a sequential test for a given error rate $\epsilon$. According to the obtained results in Fig. \ref{fig:0.1.2}, the normalized average latency of these codes is computed using Theorem \ref{th:1} when an error rate is not larger than $\epsilon = 10^{-9}$ as presented in the Table \ref{table:ch:1.4.3.1}. It can be shown that a good sequential test allows for reducing the latency. Such a scheme is particularly efficient for short blocklengths. Indeed, for a blocklength of $500$, messages can be sent with an error rate of $10^{-9}$ using $54\%$ and $62\%$ of the time symbol, respectively for a rate of $0.5$ and $0.95$ bit per channel use. 
\begin{table}[!htbp]
\centering
\renewcommand{\arraystretch}{1.3}
\caption{Normalized average latency using an optimal sequential detection with an error rate required $\epsilon = 10^{-9}$}
\begin{tabular}{|c|c|c|c|c|c|c|}
\hline
Blocklength $n$  &$150$& $300$ & $500$& $1000$&$2000$&$5000$ \\
\hline
$\frac{1}{T}\mathbb{E}\left[\tau \right] $ for $R = 0.5$ &$0.34$ & $0.46$ & $0.54$ &$0.64$ &$0.73$ &$0.82$\\
\hline
$\frac{1}{T}\mathbb{E}\left[\tau \right] $ for $R = 0.95$ &$0.54$&$ 0.41$&$0.62$&$0.71$ & $0.78$& $0.86$\\
\hline
\end{tabular}
\label{table:ch:1.4.3.1}
\end{table}

As the blocklength grows, the required SNR to reach the required performance decreases and since the error rate is close to one when $R > C$, the average latency increases. This achievable latency can be analysed for various size of the code under various channel conditions. Results are presented in Figs. \ref{fig:0.1.3}-\ref{fig:0.1.5}. For a fixed blocklength, Fig. \ref{fig:0.1.3} shows that as the rate increases, in other words the information block size grows, the time symbol required to receive messages becomes larger. Nevertheless, it is interesting to note that for a rate of $R = 0.5$ and a blocklength of $n = 5000$ symbols, only $75\%$ of the time symbol is needed to reach an error rate not larger than $\epsilon = 10^{-9}$, which can in turn reduce latency.       
\begin{figure*}[!t]
\centering
\includegraphics[width=5.0in,height=3.8in]{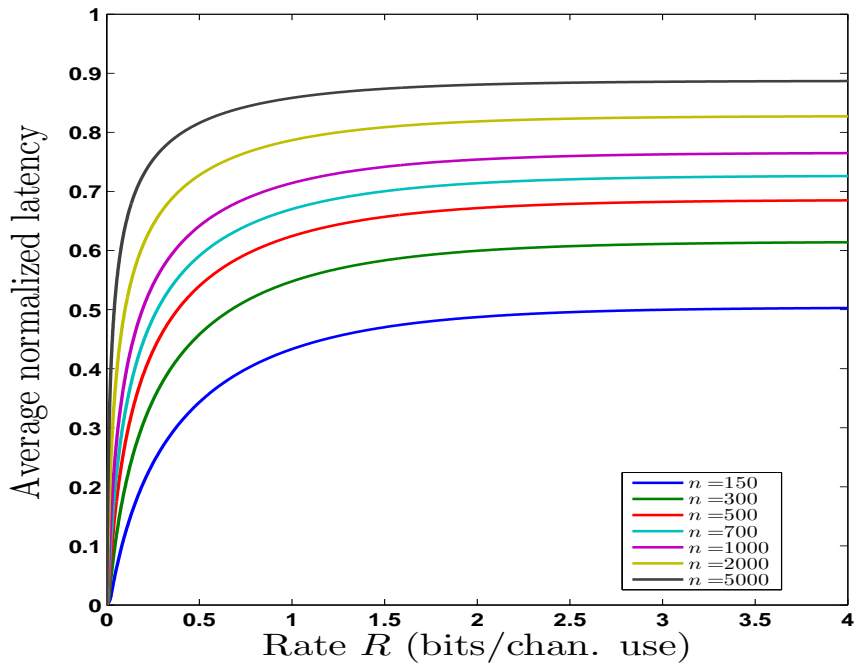}
\caption{Average latency as a function of channel code rate for various blocklength $n$ and probability of block error $\epsilon = 10^{-9}$}
\label{fig:0.1.3}
\end{figure*}

Fig. \ref{fig:0.1.4} provides the average latency for a fixed information block size. On one hand, it can be seen that the average latency decreases slightly when the code rate increases. In other words, for a fixed information block size, a large blocklength increases latency slightly. On the other hand, a large information block size increases latency. Indeed, it can be seen in this figure that for a rate of $R = 0.5$, $88\%$ of the time symbol on average is needed whereas $45\%$ of the time symbol is required for an information block size of $k = 5000$ and $k = 150$ bits respectively. These results show that the minimal latency is not only linked to the blocklength, but mostly to the information block size. 
\begin{figure*}[!t]
\centering
\includegraphics[width=5.0in,height=3.8in]{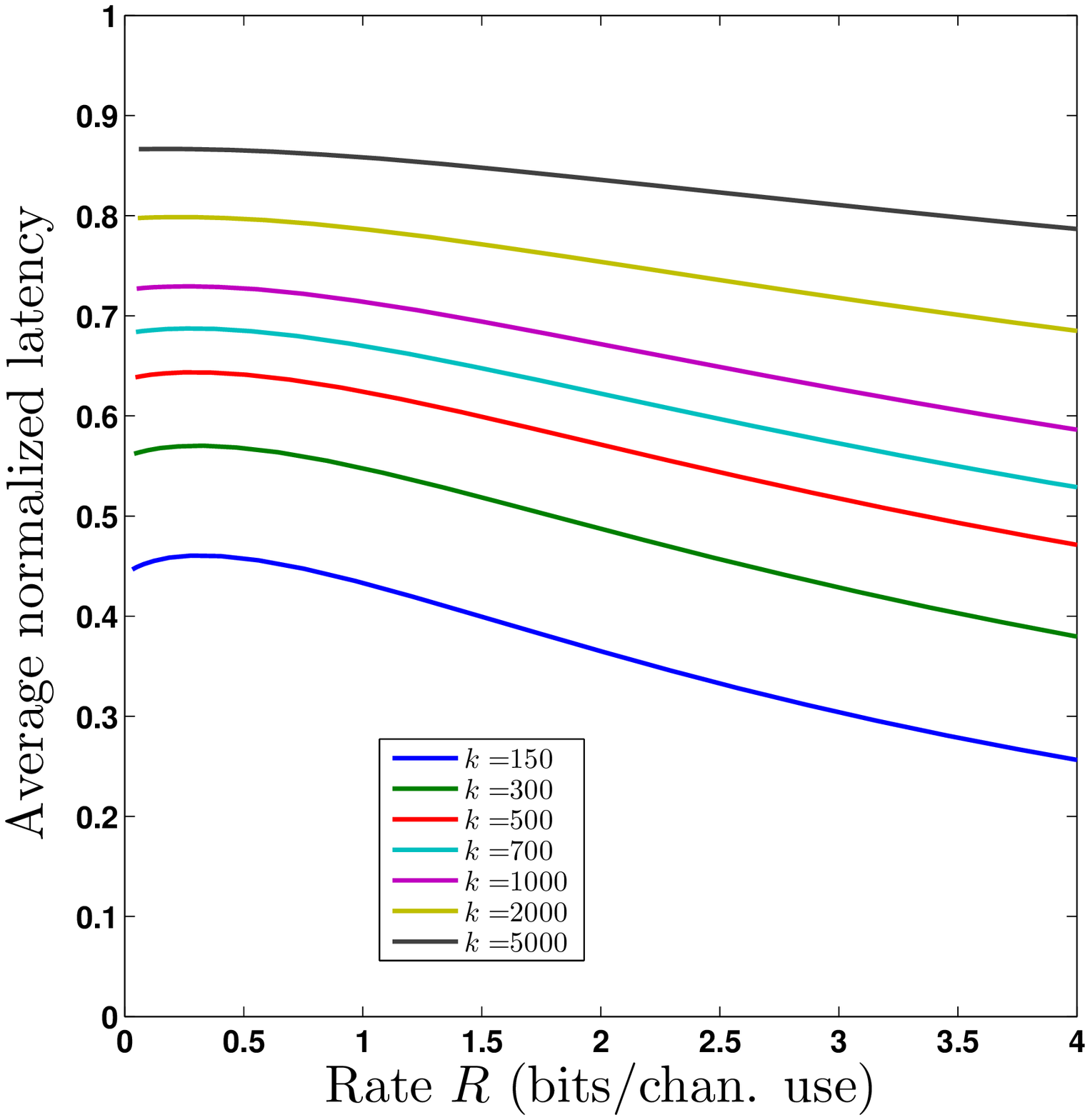}
\caption{hl{Average latency as a function of channel code rate for various information block size $k$ and probability of block error $\epsilon = 10^{-9}$}}
\label{fig:0.1.4}
\end{figure*}

The graph depicted in Fig. \ref{fig:0.1.5} shows the average latency as a function of the information block size $k$ for fixed coding rate of $R = 0.5$ under various channel conditions. It can be shown that, as the error rate is low, the average time symbol needed to decode faster, decreases because the SNR to reach the required error is high. For an information block size of $k = 1000$ bits, $69\%$ of the time symbol on average is needed whereas $84\%$ is required for an error rate of $\epsilon = 10^{-12}$ and $10^{-3}$ respectively. Furthermore, as the information block size grows, the time symbol needed to reach such an error probability increases. Indeed, for an information block size of $k \geq 10^{6}$ bits, the average time needed is close to $100\%$. Thus, there is no advantage to use early-detection schemes over a very large information block size because $100\%$ of the symbol duration is required. 

\begin{figure*}[!t]
\centering
\includegraphics[width=5.0in,height=3.8in]{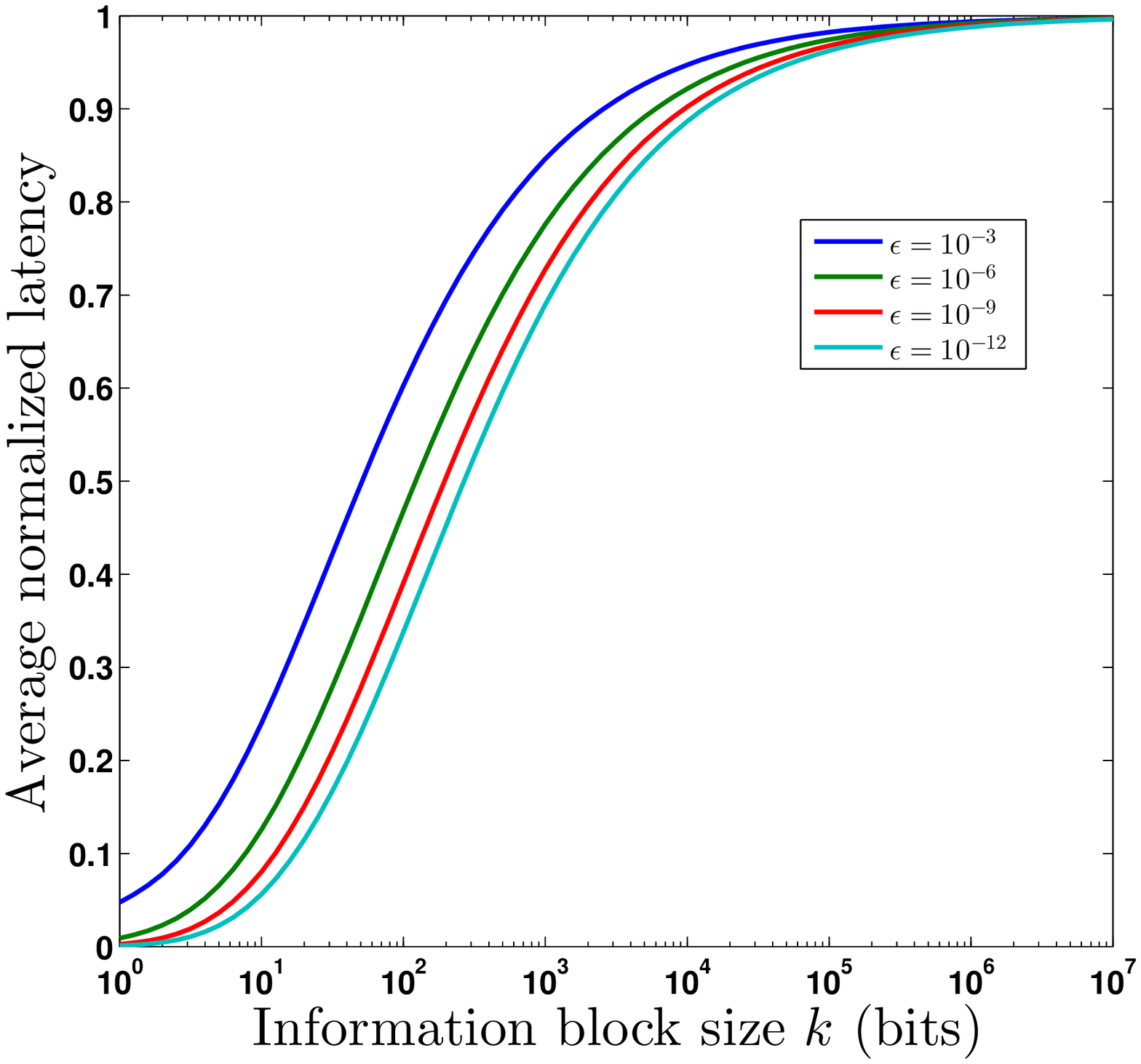}
\caption{hl{Average latency as a function information block size $k$ and rate $R = 0.5$ for various probability of block error $\epsilon$}}
\label{fig:0.1.5}
\end{figure*}

It can be shown that for both synchronous detection and early-detection schemes, the minimal latency is essentially a short message issue. Furthermore, early-detection schemes do not need to wait for the end of the symbol duration to make decisions, which can in turn reduce latency. This can be obtained by using sequential detections. In literature, one of the most popular sequential detection schemes is the multihypothesis sequential probability ratio test (MSPRT) proposed by \cite{Baum1994,Veeralli1995,Dragalin1999,Dragalin2000}. These classes of sequential detection schemes are well established, for which their asymptotic optimality have been investigated in \cite{Baum1994,Dragalin1999,Dragalin2000,Lehmann1959,Chernoff1959}. In particular, Dragalin \textit{et al.} \cite{Dragalin1999} have shown that the proposed MSPRT's are asymptotically optimal relative to the average time needed to make decisions, but also to any positive moment of the stopping time distribution. Unfortunately, the MPSRT schemes might be difficult to implement in practice because codewords consist of several number of symbols; typically more than a hundred symbols. Specifically, for an information block size of $k = 100$ bits, there is $2^{k}$ possible messages, which renders MSPRT not tractable. 

In the next sub-section, we design two sequential tests whose their structures are simple and would facilitate implementation. We provide two ways to achieve low-latency communication. The first employs a list decoder in which the MPSRT can make a decision based on few number of hypothesis. The second uses a sequential test guided by a ``genie'' aided such as error-detecting codes. In the latter, probability ratio tests are not necessarily required.

\subsection{On the Design of Efficient Sequential Detection Schemes}
\subsubsection{An MSPRT Scheme Guided by List-decoding}
Consider $2^{k}$ possible messages of $k$ bits encoded by an arbitrary ($n$,$2^{k}$,$\epsilon$) code whose each symbol has a fixed duration $T$. An encoded message $m$ in which all symbols are simultaneously transmitted through a parallel AWGN channel with $n$ branches. The MSPRT allows for reducing latency by choosing which message was transmitted among $M$ possible messages as soon as the probability of its correct detection is high enough. Inspired by previous works on sequential detection, the early-detection problem can be formulated for multidimensional signalling. By using Bayes's rule, the posterior probabilities can be written as:
\begin{equation}
\text{P}(m|\mathbf{Y}_{t}) = \frac{\pi_{m} \prod\limits_{i=1}^{t}\text{P}(\mathbf{Y}_{t}|m)}{\sum\limits_{j = 1}^{M}\pi_{j} \prod\limits_{i=1}^{t}\text{P}(\mathbf{Y}_{t}|j)} > S_{m}
\end{equation}
where $\pi_{j}$ is the prior probability of the transmitted message, $\text{P}(\mathbf{Y}_{t}|j)$ is the likelihood function for $j = 1,2\cdots, M$, and $0 \leq S_{m} \leq 1$ is a positive parameter. Hence, the stopping time $\tau_{m}$ and the decision $\delta$ is given by:
\begin{subequations}
\begin{align}
\tau_{m} = \inf\left\lbrace t: \text{P}(m|\mathbf{Y}_{t}) > S_{m} \right\rbrace \\
\delta = \hat{m},~  \text{where}~ \hat{m} = \argmax\limits_{1\leq m\leq M}\left(\pi_{m} \prod\limits_{i=1}^{t} \text{P}(\mathbf{Y}_{t}|m) \right)  
\end{align}
\label{eq:3.4.1}
\end{subequations}
Such an equation means that the receiver stops as soon as the posterior probability exceeds a threshold and decide that $m$ was transmitted. In particular, for $M = 2$, the sequential test yields the binary SPRT known as the Wald's SPRT \cite{Wald1948,Siegmund1985,Chernoff1959}. Under such a condition, if $m = 2$, then $\text{P}(m = 2|\mathbf{Y}_{t})$ is given by:  
\begin{equation}
\text{P}( 2|\mathbf{Y}_{t}) = \frac{\pi_{ 2} \prod\limits_{i=1}^{t}\text{P}(\mathbf{Y}_{t}| 2)}{\pi_{ 2} \prod\limits_{i=1}^{t}\text{P}(\mathbf{Y}_{t}| 2)+\pi_{ 1} \prod\limits_{i=1}^{t}\text{P}(\mathbf{Y}_{t}| 1)} > S
\label{eq:1.1}
\end{equation}
It follows that the likelihood ratio of the SPRT is written as:
\begin{equation}
\Lambda(t) = \frac{\prod\limits_{i=1}^{t}\text{P}(\mathbf{Y}_{t}\mid 2)}{\prod\limits_{i=1}^{t}\text{P}(\mathbf{Y}_{t}\mid 1)} > \frac{\pi_{1}}{\pi_{2}}\frac{S}{1-S}
\label{eq:1.2}
\end{equation}
in such a case, the binary SPRT test can be written as follows:
\begin{subequations}
\begin{align}
\tau_{m} = \inf\left\lbrace t: \Lambda(t) \not\in \left[A,B \right]  \right\rbrace \\
\delta = 2,~  \text{if}~ \Lambda(t)>B~\text{and}~\delta = 1,~  \text{if}~ \Lambda(t)<A
\end{align}
\label{eq:3.4.34}
\end{subequations}
where $B = \pi_{1}S/(\pi_{2}-\pi_{2}S)$ and $0 \leq A \leq B$. In literature, it has been proven that the Wald's SPRT is optimal in the sense that it minimizes the expectation of the sample size for which the probabilities of error do not exceed a predefined value \cite{Wald1948,Lehmann1959,Chernoff1959}. 

The performance of the system is given by the average of the message error probabilities:
\begin{equation}
\begin{split}
\epsilon =& \sum\limits_{m = 1}^{M}\pi_{m}\text{P}_{\mathbf{Y}\mid m}(g(\mathbf{Y}_{\tau})\neq m)\\
\end{split}
\label{eq:3.4.35}
\end{equation}
where $\text{P}_{\mathbf{Y}\mid m}(g(\mathbf{Y}_{\tau}) \neq m)$ is the probability of error when the sequential test stopped at $\tau < T$ and chose the wrong message. In \cite{Baum1994,Dragalin1999}, it has been proven that equation (\ref{eq:3.4.35}) has an upper bound for a given threshold $S_{m}$. This is given by the following theorem
\begin{theorem}[Baum and Veeravelli]
Let $\epsilon_{m',m} = \text{P}_{\mathbf{Y}\mid m'}(g(\mathbf{Y}_{\tau})= m)$ be the probability of deciding the message $m$ was transmitted when $m'$ was sent, and $\epsilon_{m} = \text{P}_{\mathbf{Y}\mid m}(g(\mathbf{Y}_{\tau})\neq m)$ the probability of having deciding incorrectly that the message $m$ was transmitted. Then $\epsilon_{m',m}$ and $\epsilon_{m}$ are upper bounded: 
\begin{subequations}
\begin{align}
\epsilon_{m} \leq \sum\limits_{\substack{m'=1 \\ m'\neq m}}^{M} \pi_{m'}\epsilon_{m',m} \leq \pi_{m} \frac{1-S_{m}}{S_{m}}\\
\epsilon = \sum\limits_{m = 1}^{M} \epsilon_{m} \leq \sum\limits_{m = 1}^{M}\pi_{m} \frac{1-S_{m}}{S_{m}}\\
\epsilon  \leq 1-S~\text{if}~ S = S_{1} = S_{2} \cdots = S_{M} 
\end{align}
\label{eq:3.4.36}
\end{subequations}
\label{th:1.3}
\end{theorem}
It can be noted that the error rate $\epsilon$ is intrinsically linked to the threshold $S_{m}$. Indeed, it can be shown that if $S_{m}$ is small, then the latency is reduced but the error increases whereas for a large value of $S_{m}$, equation (\ref{eq:3.4.1}) would never be satisfied. In this case, the decision on the current message is made at the end of the symbol duration. Therefore, this threshold must be defined such that the latency is minimized for which $\text{P}_{\mathbf{Y}\mid m}(g(\mathbf{Y}_{\tau})\neq m)$ does not exceed a predefined value. 

The average latency of early-detection schemes is then the average of the sample size used for decisions on a message $m$\setcounter{equation}{38}: 
\begin{equation}
\mathbb{E}\left[\tau \right] = \sum\limits_{m = 1}^{M} \pi_{m}\mathbb{E}_{m}\left[\tau \right]
\label{eq:3.4.37}  
\end{equation}
where $\mathbb{E}_{m}\left[\tau \right]$ denotes the expectation with respect to measure $\text{P}_{\mathbf{Y}\mid m}$. We should note that for $M = 2$, Poor and Hadjiliadis \cite{Poor2009} have proven that equation (\ref{eq:3.4.35}) can be bounded by the following proposition. 
\begin{propos}[Poor and Hadjiliadis]
Suppose that the Kullback-Leibler (KL) distances denoted by $D_{1}$ and $D_{2}$ are positive and finite under the hypotheses that $m = 1$ and $m = 2$ were sent respectively. In addition, assume that the thresholds $A$ and $B$ are respectively set such that the error probabilities for such a test are $\text{P}_{\mathbf{Y}\mid m = 1}(g(\mathbf{Y}_{\tau}) = 2) = \alpha$ and $\text{P}_{\mathbf{Y}\mid m = 2}(g(\mathbf{Y}_{\tau}) = 1) = \gamma$ . Then $\mathbb{E}_{1}\left[\tau \right]$ and $\mathbb{E}_{2}\left[\tau \right]$ are given by:
\begin{subequations}
\begin{align}
\mathbb{E}_{1}\left[\tau \right] \geq D_{1}^{-1}\left[\alpha \log\left(B\right)+ (1-\alpha) \log\left(A\right) \right] \\
\mathbb{E}_{2}\left[\tau \right] \geq D_{2}^{-1}\left[(1-\gamma) \log\left(B\right)+ (\gamma) \log\left(A\right)\right]
\end{align}
\label{eq:3.4.111}
\end{subequations}
for any $B\leq (1-\gamma)/\alpha$ and $A\geq \gamma/(1-\alpha)$. 
\label{remrk:1.1}
\end{propos}
Moreover, for $M>2$, \cite{Dragalin1999,Baum1994} have investigated an asymptotic formula of any positive moments of $\tau$, thereby establishing the following theorem.
\begin{theorem}[Dragalin \textit{et al.}]
For any minimal KL distance $D_{m}$ between the hypothesis that $m$ was transmitted and the other, for all $i \geq 1$ 
\begin{equation}
\mathbb{E}_{m}\left[\tau^{i} \right] \sim \left(-\log\left(\frac{1-S_{m}}{S_{m}} \right)D_{m}^{-1}  \right)^{i}~\text{as}~\max_{m} S_{m} \rightarrow 1  
   \end{equation}
\end{theorem}

It can be seen that, early-detection schemes using sequential probability ratio tests allow for reducing latency on average. However, two major concerns should be raised. The first, with respect to the current-state of the art, such an optimal threshold might be difficult to obtain due to various channel conditions, modulation, and coding schemes. To the best of our knowledge, it is not proven that these theorems hold over a very large number of hypotheses. In \cite{Baum1994,Dragalin2000} the number of hypotheses $M$ is not larger than four. Recent research on MSPRT have not been explored up to now. For these reasons, we will conjecture such a threshold $S_{m}$ with respect to various channel conditions as well as over a large number of possible messages. 

In addition, for a large information block size, the receiver needs $2^{k}$ tests in order to choose which message has the largest posterior probabilities. Thus, such a sequential test might not facilitate implementation when the information block size is very large. Interestingly, a list-decoder can significantly reduce the number of  hypothesis for sequential tests by providing a list of the $\ell < M$ most probable messages. Since $\ell$ is less than $M$, prior probabilities for $\ell$ most probable possible messages should be redefined as: $\bar{\pi}_{m} = \pi_{m}/(\sum_{j = 0}^{\ell} \pi_{j})$, which renders equation (\ref{eq:3.4.1}) accurate. From the Theorem \ref{th:1.3}, it follows that the threshold $S_{m}$ can be written as:
\begin{equation}
S_{m} \leq \frac{1}{1+ \bar{\pi}_{m}^{-1}\sum\limits_{\substack{m'=1 \\ m'\neq m}}^{\ell} \bar{\pi}_{m'}\epsilon_{m',m}}
\label{eq:23.4}
\end{equation}
In particular, for $\pi_{j} = 1/M~\forall~j = 1,2,\cdots M$, prior probabilities of $\ell$ most probable possible messages are defined by $\bar{\pi}_{j} = 1/\ell~\forall~j = 1,2,\cdots \ell$. It is straightforward to show that, as the number of $\ell$ most probable messages is small, $S_{m}$ grows, which could therefore decrease the error probabilities but increase the average latency and vice-versa. 

\subsubsection{A Sequential Detection Guided by Error-detecting Codes} 
A list-decoding guided by a sequential probability ratio test is a possible scheme to achieve low-latency communication. However, channel characteristics must be known by the receiver, because the threshold is mostly determined by channel conditions. Nevertheless, the latter can be estimated by using various signal processing methods, which may not facilitate implementation. 

As an alternative, most of the transmitted messages are encoded by channel codes and error-detecting codes. We could therefore exploit these codes for sequential detection problem to achieve early detection. Hence, sequential probability ratio tests are not required. With the help of a perfect error-detecting code, the receiver can make a decision as soon as such a code does not detect an error on the current signal message. We should remark that over a large number of possible messages, such a scheme can also be used with list-decoding. We believe that a list-decoding with cyclic redundancy check (CRC) is a typical structure that can reduce latency because it increases the minimum distance between messages. 

As a fundamental result, we proved that the optimal latency can be achieved by minimizing the time needed to make decisions. In practice, early-detection schemes can be designed via sequential tests based on list-decoding combined with MSPRT or error-detecting codes as depicted in Fig. \ref{fig:3.1.5}. 
\begin{figure}[!htbp]
\centering
\includegraphics[width=3.5in,height=1.2in]{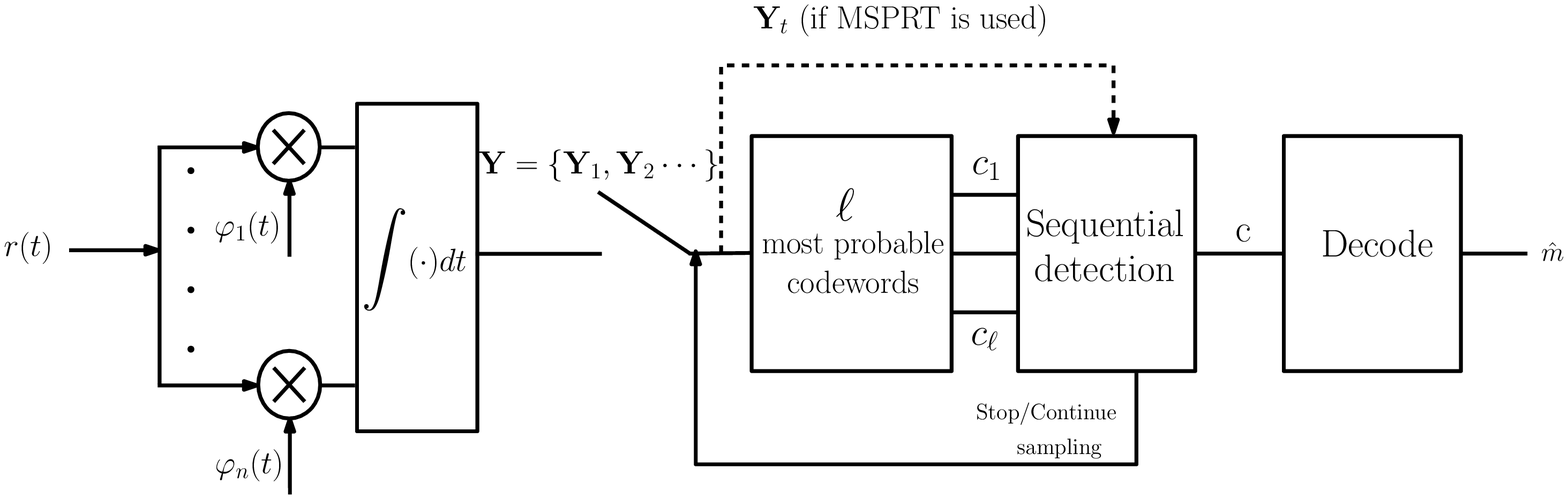}
\caption{A simple early-detection scheme based on list-decoding and sequential detection for low-latency communication}
\label{fig:3.1.5}
\end{figure}

\section{Examples of early-detection schemes for low-latency communication}
\label{sect:III}
\subsection{Low-latency Communication Under AWGN Channel}
Examples of early-detection schemes for low-latency communication are provided in this section. Noteworthy results on both latency and error probabilities are discussed under various channel conditions and message or codeword length. First, we consider a binary block message or codewords of $n$ symbols that are transmitted simultaneously through a parallel AWGN channel with $n$ branches. There are $M = 2^{n}$ possible messages whose prior probabilities are both known and uniformly distributed. Each symbol is transmitted with a fixed duration $T$ and is modulated by a binary modulation such as binary phase-shift keying (BPSK). The signal messages are denoted by $\mathbf{X}^{m} \in \mathbb{R}^{n}$ for all $m = 1,2,\cdots M$ which satisfy an input constraint as stated in equation (\ref{eq:1.1.1}).

The receiver observes a small proportion of the signal message denoted by a sequence $\mathbf{Y}_{1}$, $\mathbf{Y}_{2}, \cdots$ of independent Gaussian variables with mean $\mathbf{X}_{i}^{m}$ where $i = 1,2 \cdots$ whose satisfies equation (\ref{eq:3.bis.3}). The variance of the noise is constant and denoted by $\frac{N_{0}}{2}\mathbf{I}_{n}$. Specifically, each element of this sequence $\mathbf{Y}_{t}$ is a time-series as stated in equations (\ref{eq:3.bis.1}) and (\ref{eq:3.bis.2}). This receiver has a list-decoder which provides $\ell$ most probable messages or codewords. Then, it performs such a sequential test in order to choose the correct message quickly. For early-detection using MPSRT, it is easily verified that under an AWGN channel, such a test the takes form of equation (\ref{eq:3.4.43}).
\newcounter{mytempeqncnt2}
\begin{figure*}[!ht]
\normalsize
\setcounter{mytempeqncnt2}{\value{equation}}
\setcounter{equation}{42}
\begin{subequations}
\begin{align}
\tau_{m} = \inf\left\lbrace t: \sum\limits_{\substack{m'=1 \\ m'\neq m}}^{\ell} \exp\left(  \sum\limits_{i = 1}^{t} \frac{\left( \mathbf{X}^{m'}-\mathbf{X}^{m}\right)^{T}\mathbf{Y}_{i}}{\frac{N_{0}}{2}} \right)  < \frac{1-S_{m}}{S_{m}}  \right\rbrace \\
\delta = \hat{m},~  \text{where}~ \hat{m} = \argmin\limits_{1\leq m\leq M}\Vert\mathbf{Y}_{\tau_{m}} - \mathbf{X}^{m}\Vert 
\end{align}
\label{eq:3.4.43}
\end{subequations}
\setcounter{equation}{\value{mytempeqncnt2}}
\hrulefill
\vspace*{4pt}
\end{figure*}
\setcounter{equation}{43}

Since we have stated that each signal message satisfies an input constraint, the message error probabilities of an early detection should not be much greater than the error probabilities when the receiver makes its decision at the end of the transmitted symbol. Inspired by previous results on MSPRT and Wald's SPRT in\cite{Wald1948,Baum1994,Dragalin1999,Dragalin2000,Poor2009}, a threshold $S_{m}$ might be defined the following corollary.
\begin{coroll}
\begin{equation}
S_{m} \geq \frac{1}{1+\sum\limits_{\substack{m = 1 \\ m\neq m'}}^{\ell}P_{e}(m\rightarrow m')}
\end{equation}
where $P_{e}(m\rightarrow m')$ denotes the pairwise error probability when $\mathbf{X}^{m}$ is sent and $\mathbf{X}^{m'}$ is the only alternative. 
\end{coroll}
\begin{remk}
In the case of AWGN channels, and under normal distribution, $P_{e}(m\rightarrow m')$ is simply given by:
\begin{equation}
P_{e}(m\rightarrow m') = Q\left( \frac{\Vert\mathbf{X}^{m}-\mathbf{X}^{m'} \Vert}{\sqrt{2 N_{0}}}\right) 
\end{equation}
where $N_{0}/2$ is the variance of the noise. If we assume that all codewords can be transmitted with equiprobable prior probabilities, then we choose such a threshold $S_{m}$ such that: 
\begin{equation}
S_{m} = \frac{1}{1+\ell\sum\limits_{\substack{m = 1 \\ m\neq m'}}^{\ell}P_{e}(m\rightarrow m')}
\end{equation}
Note that when $\ell = 2$, the sequential detection yields the Wald's SPRT because there is two possible messages. For such a case, the threshold is simply given by: 
\begin{equation}
S_{m} = 1-P_{e}(m\rightarrow m')
\end{equation}
\end{remk}

By using the log-likelihood ratio, it can be shown that the stopping time rule has the following form: receiver stops as soon as:
\begin{equation}
\sum\limits_{i = 1}^{t}\mathbf{Y}_{i} > \frac{N_{0}\log(B)}{2\left(\mathbf{X}^{\ell,1}-\mathbf{X}^{\ell,2}\right)^{T}}
\end{equation}
and decide that the first most probable message was transmitted. Otherwise, choose the alternative if: 
\begin{equation}
\sum\limits_{i = 1}^{t}\mathbf{Y}_{i} < \frac{N_{0}\log(A)}{2\left(\mathbf{X}^{\ell,1}-\mathbf{X}^{\ell,2}\right)^{T}} 
\end{equation}
The thresholds $B$ and $A$ are respectively given by:
\begin{equation}
B = \frac{1-P_{e}(m\rightarrow m')}{P_{e}(m\rightarrow m')}~\text{and}~A = \frac{1}{B}
\end{equation}

\subsection{Early-detection Using MSPRT}
In this example, assume that there are $M = 1024$ possible messages of $10$ bits length modulated by a BPSK modulation, and each symbol is transmitted in parallel. The list-decoder provides a list $\ell$ most probable messages or codewords, by which the MSPRT makes the decision on the message quickly. Under various channel conditions, simulation results are presented in Table \ref{table:ch:3.1.1.1} and Fig. \ref{fig:4.1.0}. 

It can be seen that early-detection schemes using sequential tests allow for reducing the latency significantly. Compared to synchronous detection, the proposed scheme can detect messages approximately $50\%$ faster on average. As expected, by increasing the number of the $\ell$-nearest neighbours, the average latency is reduced but the cost yields an increase in the average error probabilities (see Table \ref{table:ch:3.1.1.1}). 

\begin{table}[!htbp]
\centering
\renewcommand{\arraystretch}{1.3}
\caption{Performance of early-detection schemes using list-decoding and MSPRT: sample size $T = 100$, number of messages $M = 2^{10}$, SNR $= 9.6$ decibel}
\begin{tabular}{|c|c|c|c|c|}
\hline
 & Non seq. & \multicolumn{3}{|c|}{list-decoding+MSPRT } \\
  \cline{3-5} 
& &$\ell = 2$ & $\ell = 3$& $\ell = 5$\\ 
\hline
Error $\epsilon$ &$10^{-4}$ & $10^{-4}$ & $1.2\cdot 10^{-4}$ &$1.4\cdot 10^{-4}$  \\
\hline
$\frac{1}{T}\mathbb{E}\left[\tau \right] $ &$1$&$0.56$&$0.49$&$0.46$\\
\hline
\end{tabular}
\label{table:ch:3.1.1.1}
\end{table}
\begin{figure*}[!htbp]
\centering
\includegraphics[width=5.0in,height=3.8in]{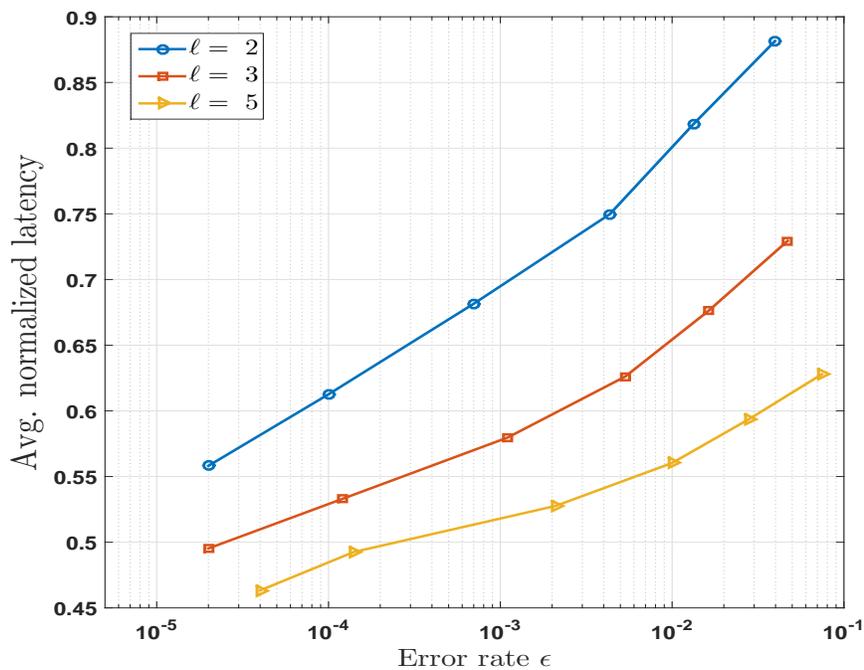}
\caption{Average latency of list-decoding+MSPRT scheme under various channel conditions: $M = 2^{10}$ messages}
\label{fig:4.1.0}
\end{figure*}

The average latency can be evaluated by increasing the number of bits in the message and a tolerated average error rate. Unfortunately, choosing the $\ell$-nearest distance between the actual signal message and $M = 2^{k}$ possible codewords can be computationally prohibitive. It should be noted that there are efficient decoding algorithms which provide the $\ell$ most probable codewords such as a list-decoding algorithm for Reed-Solomon codes \cite{Roth2000,El-Khamy2006}, polar codes \cite{Tal2011,Li2012}. 

\subsection{Early-detection Using Sequential Test Guided by Error Detection Codes}
Error-detecting codes can be used as a sequential test for which the average latency can be reduced. Consider input and output alphabets $\mathcal{A}$ and $\mathcal{B}$ and a conditional probability measure $P_{\mathbf{Y}_{t}\mid \mathbf{X}}: \mathcal{A} \rightarrow \mathcal{B}$. Assuming a perfect error-detecting code such that the decoder  $P_{\mathbf{Z}\mid \mathbf{Y}_{t}}: \mathcal{B} \rightarrow \left\lbrace 0,1\right\rbrace $ where '$0$' indicates no error in data whereas '$1$' indicates error. It can be easily seen that one can reduce latency by the following stopping rule:
\begin{subequations}
\begin{align}
\tau_{m} = \inf\left\lbrace t: P_{\mathbf{Z}\mid \mathbf{Y}_{t}}(\mathcal{B}) = 0  \right\rbrace\\
\delta = m,~  \text{where}~ m \in \mathcal{A}, 
\end{align}
\label{eq:4.1.3bb}
\end{subequations}
In other words, stop the test as soon as the decoder has not detected errors in data, and decide that the word $m$ was transmitted. CRC codes can be used as a ``genie'' guided for such a sequential test. It should be noted that we could easily define the stopping rule by verifying the received and the computed remainder. We have: stop the test as soon as the received and the computed remainder are equal. Unfortunately, if we try to detect too early in noisy channels, then their error detection performance can dramatically decrease, which could degrade the overall performance. Nevertheless, one could introduce empirically a minimum time for which CRC is effective.  

A simulation result is presented in Table \ref{table:3.3}, in which we have only considered short messages. Each symbol is modulated in BPSK with $8$ bits-CRC and $16$ bits-CRC. For a fixed symbol duration, the receiver performs the sequential test that satisfies equation (\ref{eq:4.1.3bb}). We have optimized such a scheme for which the probability of block error is both equal for synchronous detection and early-detection. It can be seen that such a sequential test reduces the average latency while maintaining a probability of block error $\epsilon = 10^{-3}$ in an AWGN channel. Moreover, a large CRC length improves the latency for these short messages. 

\begin{table}[!htbp]
\centering
\renewcommand{\arraystretch}{1.3}
\caption{Normalized average latency using sequential test guided by $n-k$ bits-CRC with a required probability of block error $\epsilon = 10^{-3}$ in AWGN channel}
\begin{tabular}{|c|c|c|c|}
\hline
Information block size $k$  &$150$& $200$ & $500$  \\
\hline
$\frac{1}{T}\mathbb{E}\left[\tau \right] $ for $8$ bits-CRC &$0.43$ & $0.45$ & $0.49$ \\
\hline
$\frac{1}{T}\mathbb{E}\left[\tau \right] $ for $16$ bits-CRC &$0.35$ & $0.356$ & $0.42$ \\
\hline
\end{tabular}
\label{table:3.3}
\end{table}

These results show that early-detection schemes can be used for reducing latency if symbols are transmitted in parallel over the channel. We have demonstrated that sequential tests such as MSPRT and sequential tests guided by ``genie'' aided such as error-detecting codes are effective in practice. 

\subsection{Does Early-detection Work in OFDM?}
\label{sect:IVb}
We have seen that early detection via sequential detection scheme reduces the latency if symbols are transmitted in parallel. OFDM signalling is typical because data is transmitted in parallel by assigning each symbol to one carrier. This is efficient only if the orthogonality of the carriers holds. The sub-carrier spacing must be proportional to the inverse of symbol duration $1/T$. It can be easily shown that the orthogonality does not hold before the end of the transmitted symbol. A critical question that needs to be explored: can we decide earlier by evaluating distances between different OFDM signals? To answer this, consider a codeword $\mathbf{X}^{m} \in \mathbb{C}^{n}$ arbitrarily chosen among $M$ possible codeword, and its OFDM signal $s_{m}(t)$ which uses the inverse of discrete Fourier transform (IDFT) such that:
\begin{equation}
  s_{m}(t) = \sum_{k = 1}^{n} X_{i}^{m}e^{\frac{j2\pi kt}{T}},
\label{eq:3.4.1aaa}
\end{equation}  
where $X^{m}_{k}$ is the information signal which is multiplied by a carrier frequency $k/T$. If it is possible to define early-detection schemes in OFDM, then the receiver can make a decision quickly as soon as the distance between $y(t)$ and $s_{m}(t)$ reaches a threshold $S_{m}$. Hence, the stopping rule can be defined as:
\begin{subequations}
\begin{align}
\tau_{m} = \inf\left\lbrace t: \Vert y(t)- s_{m}(t) \Vert < S_{m} \right\rbrace \\
\delta = \hat{m},~  \text{where}~ \hat{m} = \argmin\limits_{1\leq m\leq M}\Vert y(t)- s_{m}(t)  \Vert,
\end{align}
\label{eq:3.4.1b}
\end{subequations} 
where $y(t)$ is the received signal plus noise. There is several noteworthy feature on distances of these OFDM signals.

\begin{remk}
Assuming a random coding where $\mathbf{X}^{m}$ and $\mathbf{X}^{m'}$ are i.i.d. random vectors, the distance $s_{m}(t)$ and $s_{m'}(t)$ is given by:
\begin{equation}
\begin{split}
\left(d_{t}^{mm'}\right)^{2} &= \Vert\mathbf{X}^{m} \Vert^{2}t+\Vert\mathbf{X}^{m'} \Vert^{2}t-2\Re\left(  \int_{0}^{t} s_{m}(t) s_{m'}^{*}(t)  dt \right)\\
\end{split}
\end{equation}
where $\Re\left(\cdot\right)$ denotes the real part of a complex number. Since $\mathbf{X}^{m}$ and $\mathbf{X}^{m'}$ are independent, then the covariance is null. Hence, the square distance $\left(d_{t}^{mm'}\right)^{2}$ is approximately linear over time. As a result, it is possible to use early-detection efficiently when random coding schemes are employed.      
\end{remk}

Unfortunately, codewords are generally non i.i.d. random vectors. Consider a codeword $\mathbf{X}^{m}$ and its nearest neighbour $\mathbf{X}^{m'}$ such that the distance over time is given by:
\begin{equation}
\begin{split}
\left(d_{t}^{mm'}\right)^{2} &= \int_{0}^{t} \bigg\vert\sum\limits_{k \in \mathcal{K}}\left(X_{k}^{m}-X_{k}^{m'}\right)e^{j2\pi\frac{k}{T}t}\bigg\vert^{2} dt\\
\end{split}
\label{eq:3.4.1c}
\end{equation}
where $\mathcal{K}$ is a subset in which $X_{k}^{m}-X_{k}^{m'} \neq 0, \forall k \in \mathcal{K}$, otherwise it is equal to zero. One can find that the distances between OFDM signals are non-linear functions over time which could render early-detection schemes not efficient. This is due to dimensions that overlap each other $\forall t \in [0,T[$. For example, consider a codebook of $M$ codewords that are mapped by quadrature phase-shift keying  (QPSK) modulation (two bits in each dimension), in which there are at most two different symbols among these $n$ dimensions, \textit{i.e.} the number of elements in $\mathcal{K}$ is equal to one or two. In such a case, we observe these following remarks.

\begin{remk}
When the number of elements of the subset $\mathcal{K}$ denoted by $ \#\mathcal{K}$ is equal to one, we can find from equation (\ref{eq:3.4.1c}) that $\left(d_{t}^{mm'}\right)^{2}$ is linear. However, when $ \#\mathcal{K} = 2$, then distances over time grow linearly but a sinusoid of a frequency $(k_{1}-k_{2})/T$ has to be taken into account due to the overlapping dimension. It can be shown that Fig. \ref{fig:4.1.1} illustrates the purpose. In $\mathcal{K}$, we have modified one single bit in each dimension. When $ \#\mathcal{K}$ increases, there is a superposition of multiple sinusoids of a frequency $(k_{i}-k_{j})/T~\forall k_{i}$ and $\neq k_{j} \in \mathcal{K}$ which tends to linearize distances over time. Hence, it is equivalent to use a random coding scheme.
\end{remk}

\begin{figure*}[!htbp]
\centering
\includegraphics[width=5.0in,height=3.8in]{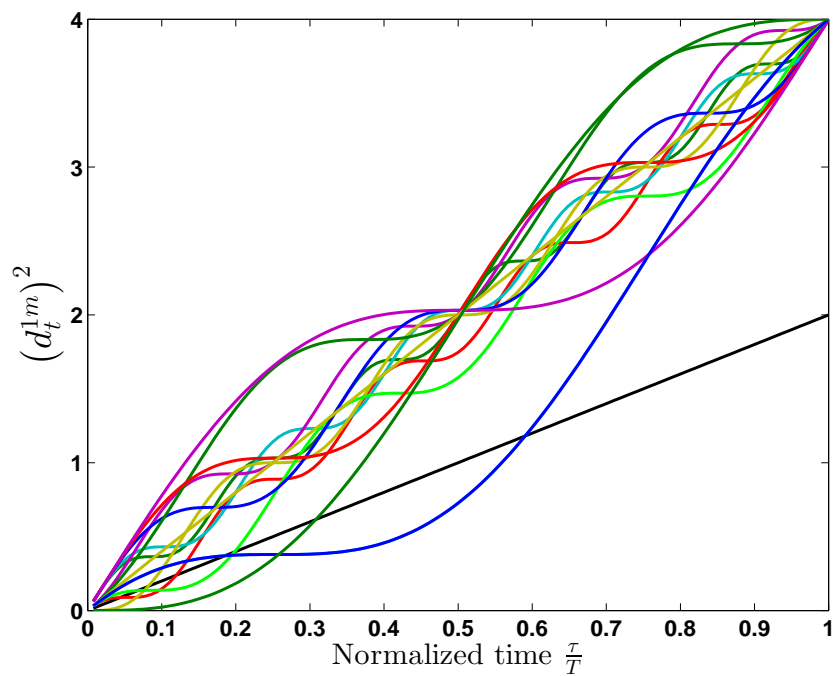}
\caption{Distances over time between OFDM signals. Number of sub-carrier: $128$. Codewords mapped by QPSK modulation. Number of different symbols in codewords $\#\mathcal{K} = 1$ represented by the black curve; number of different symbols in codewords $ \#\mathcal{K} = 2$ represented by colored curves}
\label{fig:4.1.1}
\end{figure*}

\begin{remk}
Latency can be reduced through early-detection schemes if $\left(d_{t}^{mm'}\right)^{2}$ is linear for all $m \neq m'$ in an arbitrary codebook. To do so, it is possible to use pre-coding random rotation matrices to linearize these distances. Assuming a matrix $\mathbf{H}$ whose angles are i.i.d., and considering $A_{l}^{mm'} = X_{l}^{m}-X_{l}^{m'}$, equation (\ref{eq:3.4.1c}) can be rewritten as:
\begin{equation}
\begin{split}
\left(d_{t}^{mm'}\right)^{2} &= \int_{0}^{t} \bigg\vert\sum\limits_{k \in \mathcal{K}}\sum\limits_{l}H_{k,l}A_{l}^{mm'}e^{j2\pi\frac{k}{T}t}\bigg\vert^{2} dt\\
\end{split}
\label{eq:3.4.1e}
\end{equation}
where $H_{k,l}$ is an element of $\mathbf{H}$. Assuming $\mathcal{K} = \left\lbrace k_{1},k_{2} \right\rbrace$, we obtain equation (\ref{eq:3.4.1d}) where $H_{k_{2},l}^{*}$ is the complex conjugate of the element $H_{k_{2},l}$. Since we have stated that $\mathbf{H}$ is a random rotation matrix in which elements are i.i.d., then $\sum_{l} H_{k_{1},l} H_{k_{2},l}^{*}$ must be zero. Therefore, the squared distance $\left(d_{t}^{mm'}\right)^{2}$ must be linear. It should be noted that are orthogonal matrices fulfils such a condition.
\newcounter{mytempeqncnt3}
\begin{figure*}[!ht]
\normalsize
\setcounter{mytempeqncnt3}{\value{equation}}
\setcounter{equation}{55}
\begin{equation}
\begin{split}
\left(d_{t}^{mm'}\right)^{2} &= \int_{0}^{t} \bigg\vert\sum\limits_{l} H_{k_{1},l}A_{l}^{mm'}e^{j2\pi\frac{k_{1}}{T}t}+ \sum\limits_{l}H_{k_{2},l}A_{l}^{mm'}e^{j2\pi\frac{k_{2}}{T}t}\bigg\vert^{2} dt\\
&=  \left(\sum\limits_{l} \bigg\vert H_{k_{1},l}A_{l}^{mm'} \bigg\vert^{2} + \sum\limits_{l}\bigg\vert H_{k_{2},l}A_{l}^{mm'} \bigg\vert^{2}\right)t+2\Re\left(\int_{0}^{t} \sum\limits_{l} H_{k_{1},l} H_{k_{2},l}^{*}\bigg\vert A_{l}^{mm'} \bigg\vert^{2} e^{j2\pi\frac{k_{1}-k_{2}}{T}t}dt \right)
\end{split}
\label{eq:3.4.1d}
\end{equation}
\setcounter{equation}{\value{mytempeqncnt3}}
\hrulefill
\vspace*{4pt}
\end{figure*}
\setcounter{equation}{56}
\end{remk}

A Hadamard matrix is a typical example in which $\left(d_{t}^{mm'}\right)^{2}$ can be linear. By taking results obtained in Fig. \ref{fig:4.1.1}, we apply a complex-valued Hadamard orthogonal matrix $\mathbf{H}$ of $128 \times 128$ in these codewords mapped by QPSK. Fig. \ref{fig:4.1.2} shows that  $\left(d_{t}^{mm'}\right)^{2}$ are approximately linear which could render early-detection schemes efficient by minimizing the time required to make a decision. These results show that there is evidence that latency can be reduced with OFDM signalling or any orthogonal bases by using early-detection schemes. Specifically, we proved that messages can be detected quickly by using random coding schemes and pre-coding orthogonal matrices. 

It should be noted that better orthogonal or non-orthogonal bases whose maximum distances over time could be employed for reducing latency. For example, recently, new waveforms for the $5^{\mbox{th}}$ generation wireless systems have been investigated for low-latency communications such as generalized frequency-division multiplexing (GFDM), filter bank multicarrier (FBMC), etc. \cite{Fettweis2009,Farhang-Boroujeny2011,Gaspar2013,Wunder2014,Michailow2014}. The distances $\left(d_{t}^{mm'}\right)^{2}$ for these schemes should be determined.

\begin{figure*}[!ht]
\centering
\includegraphics[width=5.0in,height=3.8in]{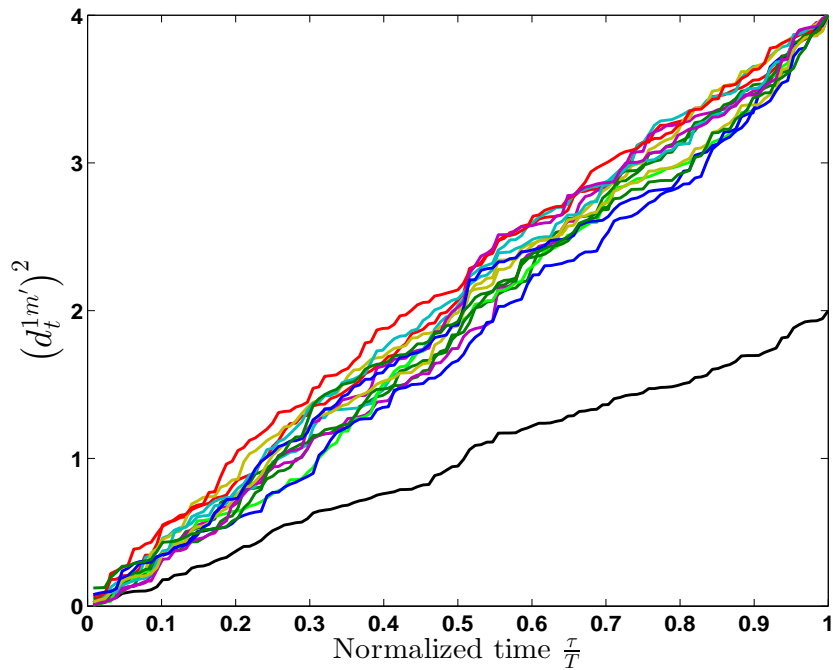}
\caption{Distances over time between OFDM signals where codewords have been pre-coded by a Hadamard orthogonal matrix. Number of different symbols in codewords $\#\mathcal{K} = 1$ represented by the black curve; number of different symbols in codewords $ \#\mathcal{K} = 2$ represented by colored curves}
\label{fig:4.1.2}
\end{figure*}

\section{Optimal latency in multi-hop systems}
\label{sect:IV}
In this work, we will focus on the serial channel relay which is a specific case of the general relay channel. This has been extensively studied for which capacity equations were obtained in  \cite{Cover1979,Kramer2005}. Unfortunately, the latency issue has not been investigated from an information theoretic perspective. In this section, we explore low-latency strategies for relays focusing on AF and DF relaying schemes. The first and the second paragraphs discuss the minimal achievable latency for synchronous-detection schemes. The third emphasizes on parallel channel relay, for which latency can be improved via early-detection schemes.

\subsection{Minimal Latency in AF Relaying Scheme} 
\begin{theorem}
For the AWGN channel with zero mean and unit variance, under equal power, bandwidth and noise density for all hops, the minimal achievable latency in multi-hop systems using AF relaying schemes in the finite-blocklength regime is the solution of equation (\ref{eq:II.B.2biss}) where the overall signal to noise ratio $\rho$ is given by:
\begin{equation}
\rho = G^{h-1}P\frac{(1-G)}{1-G^{h}}
\label{eq:3.5.1}
\end{equation}
where $G = P/(P+1)$ is the transmitted power gain which maintain the output power of a relay to $P$ when the received power is $P+1$ (transmitted power plus noise), and $h$ is the number of hops.
\label{th:5.1}
\end{theorem}

\begin{proof}
For two-hop using AF relaying schemes in the finite-blocklength regime, the power received at the first relay is $P+1$ (transmitted power plus noise). Therefore, the received power is $G(P+1)+1$. Thus, the overall signal to noise ratio $\rho$ is given by:
\begin{equation}
\rho = \frac{GP}{G+1}
\end{equation}
Under such a condition, one can find the minimal achievable latency from equation (\ref{eq:II.B.2biss}).

We generalize for multi-hop systems using AF relaying schemes in the finite-blocklength regime. Assuming $h$ hops, the received power at the second relay $G(P+1)+1$. the received power at the third relay $G(G(P+1)+1)+1$ and so on. Under equal power, bandwidth and noise density for all hops one can find that the transmitted signal gain is $G^{h-1}P$ and the noise is $1+G+G^{2}+\cdots G^{h-1}$. Hence the SNR $\rho$ is given by:
\begin{equation}
\begin{split}
\rho =& \frac{G^{h-1}P}{1+G+G^{2}+\cdots G^{h-1}}\\
\end{split}
\end{equation}
For $P > 0$, we have $G = P/(P+1)< 1$. Thus, by using the sum of the first $h-1$ terms of a geometric series, we obtain equation (\ref{eq:3.5.1}). In the finite-blocklength regime, one can find the minimal achievable latency using equation (\ref{eq:II.B.2biss}). This concludes the proof. 
\end{proof}

\begin{remk}
For a given information block size, it can be seen that as the number of hops increases, additional channel symbols are needed to maintain the required error probability $\epsilon$. In other words, the number of hops increases latency in multi-hop communication when AF relaying schemes are used.  
\end{remk}

\subsection{Minimal Latency in DF Relaying Scheme}
\begin{theorem} For the AWGN channel, since relays decode the message and forward to the destination, the minimal latency in multi-hop systems using DF relaying is simply equal to the latency of an individual hop times the number of hops:
\begin{equation}
 L_{0} = Lh
 \label{eq:3.5.2bis} 
\end{equation}
where $L$ is the minimal achievable latency given by the solution of equation (\ref{eq:II.B.2biss}) in the finite-blocklength regime.

For relays that can transmit while receiving, one can minimize latency by dividing the message in $q$ parts. Thus, the overall latency is given by:
\begin{equation}
 L_{0} = L\left( 1+\frac{h-1}{q}\right)
\label{eq:3.5.2} 
 \end{equation} 
where $h-1$ is the number of hops. 
\end{theorem}

\begin{proof}
For a message with latency $L$, we divide a message into $q$ smaller parts. Thus, the latency of each smaller message is $L/q$ which can be determined by equation (\ref{eq:II.B.2biss}). Considering $h$ number of hops, the overall latency is given by $L/q(q+h-1)$ which is the equation (\ref{eq:3.5.2}). This conclude the proof.
\end{proof}

By dividing the message into smaller parts, let us first show the overall latency with messages divided into equal parts for a two hop scheme using the achievability bound in equation (\ref{eq:II.B.2biss}). In Fig. \ref{fig:5.1.0}, we depict two cases in which the SNR is equal to $-10$ dB and $10$ dB. It can be seen that the latency can be reduced if the length of the message is directly proportional to the length of the codeword. Moreover, dividing by more than four does not reduce the latency significantly. It breaks down at about a thousand bits as shown in section \ref{sect:II}. For shorter message lengths smaller than $40$ to $200$ bits, dividing messages by two or four increase latency in multi-hop communication. For short messages, dividing messages breaks down the utility. Moreover, the absence of possible good short codes makes the division into smaller parts less attractive. Hence, the latency issue in multi-hop communication is essentially a short message issue.

\begin{figure*}[!htbp]
\centering
\includegraphics[width=5.0in,height=3.8in]{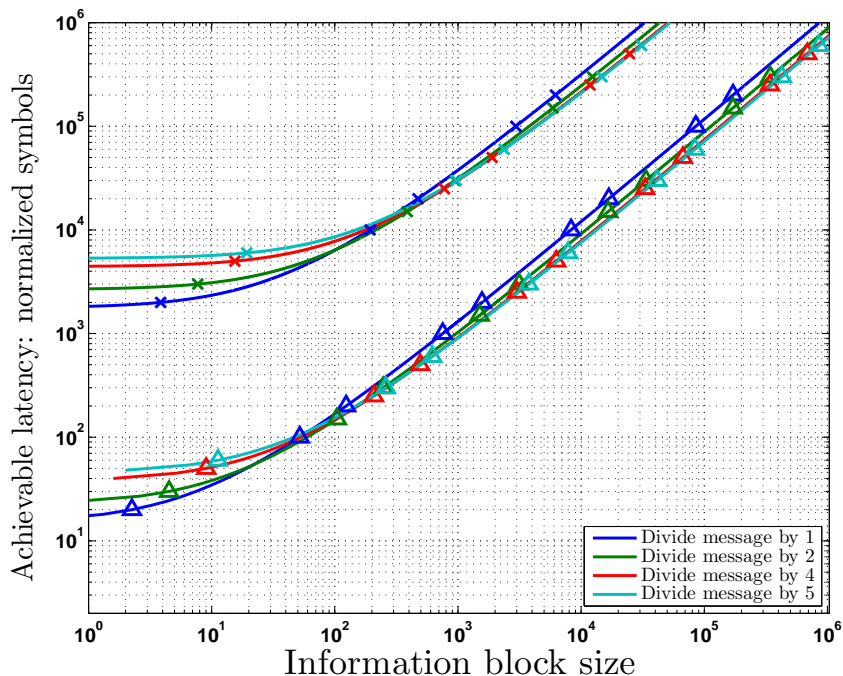}
\caption{Achievable latency in two-hop communication by dividing the message into smaller parts. Probability of block error  $\epsilon = 10^{-7}$. Cross marker SNR $= -10$ dB, triangle marker SNR $= 10$ dB}
\label{fig:5.1.0}
\end{figure*}

This result suggests that the latency issue needs to be refined for short messages. If one can use an AF scheme before the end of the symbol duration and pre-process the signal message at a relay node, the latency issue for short messages in multi-hop communication can be resolved differently than dividing messages into smaller part and use DF relaying scheme.

For long messages greater than $1$ Mbits, one can serialize and transmit over many symbols. OFDM symbols using about $2048$ dimensions is typical. However, a great number of dimensions can takes the same average time as synchronous detection (detection at the end of the symbol duration). Thus, one can minimize latency by getting a maximum capacity and break up the messages to be coded into about $1000$ channel symbols.
 
\subsection{On Minimal Latency in Multi-hop Systems via Early Detection Schemes}
In this paragraph, we explore the minimal achievable latency in two multi-hop relays when early-detection techniques are employed. First, as a general form, we consider that signals can be written as an orthogonal expansion as defined in equation (\ref{eq:1.1.00}). Specifically, a set of vectors $\boldsymbol \varphi = [\varphi_{1}(t),\varphi_{1}(t) \cdots \varphi_{n}(t)]$ forms an orthonormal basis for an $n$-dimensional space. We assume that the transmitted power is $P$ and the AWGN channel has zero mean and unit variance. The relay can sample at any instant of time to perform a sequential test such that its input $\mathbf{Y}_{t}$ at a given instant is a vector, in which each element $Y_{t,i}$ is given by: 
\begin{equation}
Y_{t,i} = \int_{0}^{t} \left( \sum\limits_{i} X_{i}^{m}\varphi_{i}(t) + n(t)\right) \varphi_{i}^{*}(t)  dt 
\end{equation}
where $X_{i}^{m}$ is a resulting symbol of the codeword $m$ in a codebook. In particular, $\Vert\mathbf{X}^{m}\Vert^{2}= nPT$. $n(t)$ denotes the additive Gaussian noise. Next, the basic relay scheme needs to be reconsidered. In its most general form, a relay has an input $\mathbf{Y}_{R}$ and an output such that:
\begin{equation}
\mathbf{X}^{m}_{R} = f_{R}(\mathbf{Y}_{R})
  \end{equation}  
where $f_{R}( \cdot )$ is a stochastic function and $\mathbf{Y}_{R} = \left\lbrace \mathbf{Y}_{t},\mathbf{Y}_{t-1} \cdots \mathbf{Y}_{1} \right\rbrace$ is the information that has been accumulated at a given relay. $\mathbf{Y}_{t}$ has been obtained by minimizing the distance at a given instant time. 

For short messages, one can apply equations (\ref{eq:3.3.5}) and (\ref{eq:3.5.2bis}) to obtain the average latency in DF schemes. Nevertheless, early-detection is not as efficient when the message is divided into smaller parts, because the receiver is required to wait until the last divided message is received before the detection of the whole message.

If the relay can receive while transmitting, then the signal message can be amplified and forwarded. And as soon as the relay can make a decision on the message, then it is possible to pre-process the information signal in order to improve latency and/or performance of multi-hop relaying systems. Indeed, by considering that $\mathbf{Y}_{R}$ has been amplified by a gain $\sqrt{G} = \left( P/(P+1)\right)^{1/2}$, it can be seen that $\Vert \mathbf{X}^{m}_{R} \Vert^{2} \leq nPT$. It is interesting to note that if the relay can make its decision before the end of the transmitted symbol, then it is possible to add an additional component $\mathbf{C}$ such that the distance between codewords can be maximized, which could therefore improve latency in the next hop. As a result, we formulate the low-latency problem for AF relaying scheme in multi-hop as the following optimization problem:
\begin{equation*}
\begin{aligned}
& \underset{m\neq m'}{\text{maximize}}
& & \Vert \mathbf{X}^{m}_{R} -\mathbf{X}^{m'} \Vert\\
& \text{subject to}
& & \mathbf{X}^{m}_{R}  = \sqrt{G}\mathbf{Y}_{R}+\mathbf{C},\\
&&& \Vert \mathbf{X}^{m}_{R}\Vert^{2} = nPT. \\
\end{aligned}
\end{equation*}
 
For binary signalling, such a case is trivial. Consider a noise vector which reduces the distance between the received signal $\sqrt{G}\mathbf{Y}_{R}$ and the codeword $\mathbf{X}^{m'}$ when the actual $\mathbf{X}^{m}$ has been transmitted. If the relay can make its decision before the end of the transmitted symbol, we should add a component $\mathbf{C}$ orthogonal to the noise component such that the optimization problem defined above is satisfied. Fig. \ref{fig:5.1.1} illustrates such a low-latency strategy for AF relaying scheme. The immediate consequence is that it is not necessary to add additional channel symbols which could thus enhance the latency as well as its reliability. In such a condition, this strategy is the optimal approach for multi-hop AF relaying schemes. 

\begin{figure}[!htbp]
\centering
\includegraphics[width=2.5in,height=2.2in]{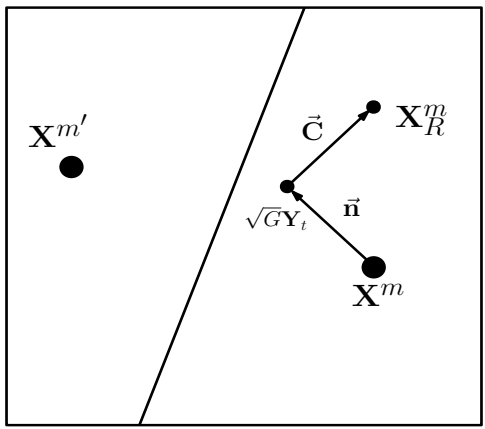}
\caption{Low-latency strategy for AF relaying scheme in multi-hop: the binary signalling case}
\label{fig:5.1.1}
\end{figure}

\section{Conclusion}
\label{sect:V}
This paper studied the optimal latency of communications. We have defined several optimal low-latency strategies for which many applications require extremely low latency and ultra reliable communication. By modifying the achievability bound of every ($n$,$M$,$\epsilon$) channel code, the minimal achievable latency in the finite blocklength regime can be derived for both synchronous detection and early-detection schemes. For synchronous detection (decision at a fixed time-period), the achievability bound gives the minimal number of symbols for a given $\epsilon$ for various fixed coding rate. On the other hand, early-detection schemes attain the optimal achievable latency by minimizing the time required to make decisions. It can be shown that such a strategy is optimal in terms of latency for fixed rate coding with no channel feedback. 

In practice, the minimal latency can be achieved by using sequential tests for which the probabilities of error do not exceed a predefined probability of block error. Such a scheme would be effective only if all symbols are simultaneously transmitted in parallel over the channel; notably, we note that OFDM-like signals transmit data in parallel by assigning each symbol to one carrier. We have developed two sequential tests based on MSPRT and early detection using CRC codes. Results show that in an AWGN channel, receivers can make decisions faster which reduce the latency while maintaining the required error probability.

For OFDM signals, it has been shown that it is possible to define early-detection schemes such that the time required to make decisions is reduced by setting stopping rules: as soon as the minimum distance between signals reaches a threshold, make a decision. However, distances over time are non-linear due to the loss of the orthogonality. In fact, since the minimum sub-carrier spacing is $1/T$, early-detection overlaps dimension of the orthonormal basis, which could produce additional false decision. Fortunately, we proved that there exists specific coding and pre-coding schemes such that these distances over time can be linearized. Specifically, we showed that random coding, random rotation, and orthogonal pre-coding matrices are typical examples. These allow for quickest detection of messages while maintaining spectral efficiency of OFDM signals.  

Furthermore, we have explored low-latency strategies for multi-hop communications. Focusing on multi-hop using AF and DF relaying schemes, we have first derived the minimal achievable latency for both synchronous detections and early-detections. Again, we proved that for short messages, early-detection schemes beats synchronous detections in DF relaying schemes. However, for long messages, such a scheme is not efficient because the receiver is required to wait until the last divided message is received before the detection of the whole message. For AF relaying schemes, early-detection can do better. Indeed, we proved that there exists an optimal scheme using AF with the help of early-detection. If relays can transmit while receiving and if they could make a decision before the end of the symbol duration, one could pre-process the signal immediately after amplify and forward. Therefore, no additional channel symbols are required which could improve the latency of AF relaying schemes.

\section*{Acknowledgment}
The authors would like to thank the following people for fruitful discussions and advices, Alexandre J. Raymond, Marwan Kanaan, Ioannis Psaromiligkos and Sergey Loyka.

\ifCLASSOPTIONcaptionsoff
  \newpage
\fi



%
\bibliographystyle{IEEEtran}
\bibliography{Data_base_latency}

\end{document}